\documentclass[12pt]{article}
\usepackage{amsmath}
\usepackage{graphicx,psfrag,epsf}
\usepackage{enumerate}
\usepackage{natbib}
\usepackage{authblk}
\usepackage{url} 
\usepackage[latin9]{inputenc}
\usepackage{float}
\usepackage{mathtools}
\usepackage{amsmath,amsthm}
\usepackage{amssymb}
\usepackage{setspace}
\usepackage{esint}
\usepackage{appendix}
\usepackage{cleveref}
\makeatletter

\providecommand{\tabularnewline}{\\}
\floatstyle{ruled}
\newfloat{algorithm}{tbp}{loa}
\providecommand{\algorithmname}{Algorithm}
\floatname{algorithm}{\protect\algorithmname}

\usepackage{latexsym}
\usepackage{amsthm}\usepackage{amsfonts}\usepackage{graphicx}\usepackage{epsfig}
\usepackage{bm}
\newcommand*{\patchAmsMathEnvironmentForLineno}[1]{%
      \expandafter\let\csname old#1\expandafter\endcsname\csname #1\endcsname
      \expandafter\let\csname oldend#1\expandafter\endcsname\csname end#1\endcsname
      \renewenvironment{#1}%
         {\linenomath\csname old#1\endcsname}%
         {\csname oldend#1\endcsname\endlinenomath}}%
    \newcommand*{\patchBothAmsMathEnvironmentsForLineno}[1]{%
      \patchAmsMathEnvironmentForLineno{#1}%
      \patchAmsMathEnvironmentForLineno{#1*}}%
    \AtBeginDocument{%
    \patchBothAmsMathEnvironmentsForLineno{equation}%
    \patchBothAmsMathEnvironmentsForLineno{align}%
    \patchBothAmsMathEnvironmentsForLineno{flalign}%
    \patchBothAmsMathEnvironmentsForLineno{alignat}%
    \patchBothAmsMathEnvironmentsForLineno{gather}%
    \patchBothAmsMathEnvironmentsForLineno{multline}%
    }
\usepackage[mathlines,displaymath]{lineno}

\include{def}
\def\dispmuskip{\thinmuskip= 3mu plus 0mu minus 2mu \medmuskip=  4mu plus 2mu minus 2mu \thickmuskip=5mu plus 5mu minus 2mu}
\def\textmuskip{\thinmuskip= 0mu                    \medmuskip=  1mu plus 1mu minus 1mu \thickmuskip=2mu plus 3mu minus 1mu}
\def\beq{\dispmuskip\begin{equation}}    \def\eeq{\end{equation}\textmuskip}
\def\beqn{\dispmuskip\begin{displaymath}}\def\eeqn{\end{displaymath}\textmuskip}
\def\bea{\dispmuskip\begin{eqnarray}}    \def\eea{\end{eqnarray}\textmuskip}
\def\bean{\dispmuskip\begin{eqnarray*}}  \def\eean{\end{eqnarray*}\textmuskip}

\newtheorem{theorem}{Theorem}

\newtheorem{lemma}{Lemma}

\newtheorem{assumption}{Assumption} 

\newcommand{\wh}{\widehat}
\newcommand{\wt}{\widetilde}
\newcommand{\ov}{\overline}
\newcommand{\bs}{\boldsymbol}
\def\blind{0}

\newcommand{\dlq}{\lq\lq}
\newcommand{\drq}{\rq\rq~}

\usepackage[mathlines,displaymath]{lineno}

\makeatother
\def\prime{\text{\tiny {\it {T}}}}
\makeatletter
\def\@seccntformat#1{\@ifundefined{#1@cntformat}%
   {\csname the#1\endcsname\quad}  
   {\csname #1@cntformat\endcsname}
}
\let\oldappendix\appendix 
\renewcommand\appendix{%
    \oldappendix
    \newcommand{\section@cntformat}{\appendixname~\thesection\quad}
}
\makeatother

\usepackage{xr}
\begin{document}
\def\spacingset#1{\renewcommand{\baselinestretch}%
{#1}\small\normalsize} \spacingset{1}


\if0\blind
{
  \title{\bf Efficient Bayesian estimation for flexible panel models for multivariate
outcomes: Impact of life events on mental health and excessive alcohol
consumption\thanks{The research of Gunawan and Kohn was partially supported
by the ARC Center of Excellence grant CE140100049. The research of all the authors was also partially supported by ARC Discovery Grant DP150104630. }  }
  \author[1]{David G. Gunawan}
    \author[1]{Chris K. Carter}
    \author[1]{Denzil G. Fiebig}
    \author[1]{Robert J. Kohn}
\affil[1]{School of Economics, University of New South Wales}
  \maketitle
} \fi

\if1\blind
{
  \bigskip
  \bigskip
  \bigskip
  \begin{center}
    {\LARGE\bf Efficient Bayesian estimation for flexible panel models for multivariate
outcomes: Impact of life events on mental health and excessive alcohol
consumption}
\end{center}
  \medskip
} \fi
\bigskip
\begin{abstract}
We consider the problem of estimating a flexible multivariate
longitudinal panel data model whose outcomes can be a combination of
discrete and continuous variables and whose dependence structures are modelled using copulas. This is a challenging problem because
the likelihood is usually analytically intractable.
Our article makes both a methodological contribution
as well as  a substantive contribution to the application. The methodological contribution is to
introduce into the panel data literature a particle Metropolis within Gibbs method to carry out Bayesian
inference, using a Hamiltonian Monte Carlo \citep{Neal:2011}
proposal for sampling the high dimensional vector of unknown parameters.
Our second contribution is to apply our method to analyse
 the impact of serious life events
on mental health and excessive alcohol consumption.
\end{abstract}

\noindent%
{\it Keywords:} Copula; Hamiltonian Monte Carlo; Particle Gibbs; Pseudo marginal method
\vfill

\spacingset{1.45} 

\section{Introduction}\label{S: introduction}
Our article considers estimating a flexible
longitudinal panel data model with multivariate
outcomes that are a combination of discrete and continuous variables.
In general, estimating such nonlinear and non-Gaussian longitudinal  models is challenging because
the likelihood is an integral over the latent individual random effects and the observations are not Gaussian.
Our article makes two substantive contributions.
First, we introduce into the longitudinal panel data methodology a version of particle Metropolis-within-Gibbs (PMwG) \citep{Andrieu:2010} that allows us to carry out Bayesian inference where the unknown model parameters are generated using a proposal obtained by Hamiltonian Monte Carlo \citep{Neal:2011}. We note that the parameter vector in panel data models is often high dimensional,
usually because there are many covariates, so that a Metropolis-Hastings
proposal based  Hamiltonian  Monte Carlo (HMC) can be much more efficient than competitor proposals such as a random walk; see
\Cref{Sec: PMCMC samplers} for more details.

We show in a simulation study that our PG approach outperforms
two other approaches for estimating such panel data models with random effects and intractable likelihoods.
The first is the standard data augmentation MCMC as in \cite{Albert1993}.
The second
is the pseudo marginal Metropolis-Hastings (PMMH) approach in \cite{Andrieu:2009}  and  \cite{Andrieu:2010}.

The motivation for our methodological development, and our second contribution,
is to investigate
the impact of life events on mental health and excessive alcohol consumption
using the Household, Income, and Labour Dynamics in Australia (HILDA) panel data set.
In the literature that investigates the impact of life events it is standard to consider
a single outcome of interest, e.g. mental health or life satisfaction,
and if multiple outcomes are considered, then  they are estimated as separate
models; see, for example, \citet{Lindeboom:2002}, \citet{Frijters:2011}
and \citet{Buddelmeyer:2016}. We contend that in many cases, we can gain additional insight
by jointly estimating the models for the outcomes.
Although it is unsurprising that there is an association between mental
health problems and excessive alcohol consumption,
establishing the nature of those links requires further research.
A natural approach is to attempt to identify causal effects as in
\citet{Mentzakis:2015}, but this relies on the availability of instruments
or a natural experiment. We use a reduced form approach
but  argue that it has the potential to provide complementary
and useful evidence; see \citet{Kleinberg:2015} for similar arguments.
Providing a description of the joint distribution of related outcomes
is often of interest and has the potential to inform about causal
links, albeit indirectly.

It is common in this literature to simplify the outcomes to allow such joint estimation. For example,
\citet{Contoyannis:2004} consider a joint model of a range of lifestyle
variables in their study of health and lifestyle choice. \citet{Buchmueller:2013}
model choices across several insurance types. Both of these studies
are representative of the methodology where outcomes of interest
are a mix of continuous and various types of discrete variables which
are all converted into binary outcomes in order to accommodate estimation
as a multivariate probit model (MVP) model using simulated maximum likelihood (SML).

The bivariate probit that accommodates the longitudinal nature
of the data serves as our  baseline model. This model is also be of independent
interest motivated by numerous applications of the MVP model; see
for example \citet{Atella:2004} and \citet{Mullahy2016}.
However, in general, there can be associated costs in simplifying the multivariate
structure to fit into the MVP framework. Abstracting from the
mix of outcomes can obscure interesting features of the relationship.
Furthermore, the MVP model imposes a linear correlation
structure that may lead to misspecification risk when the existence
of an asymmetric or nonlinear dependence structure is plausible.
Because the MVP model is limited in these ways, we
also consider models  that accommodate one continuous outcome and one that
is categorical. Here we maintain the assumption of normality for the
marginal distributions while allowing both normal and  non-normal dependence of
error terms, where non-normal dependence is obtained by using copulas.
Problems induced by correlated individual
effects are addressed with Mundlak type specifications that perform
reasonably well in practice; see for example \citet{Contoyannis:2004b}
and \citet{Woolridge:2005}.

The paper is organised as follows. The panel data model is outlined in Section \ref{sec:General-Panel-Data}. The Bayesian estimation methodology is given in Section \ref{Sec: PMCMC samplers}. The HILDA dataset is described in Section \ref{S: data and characterization}. Section \ref{S: results and discussion} discusses the estimation results. Section \ref{S: conclusions} concludes. The paper has two appendices. Appendix \ref{sec:lemma1} provides some proofs of theorem defined  in Section \ref{Sec: PMCMC samplers}. Appendix \ref{sec:Empirical-Results} presents some additional empirical results. The paper also has an online supplement whose sections are denoted as Sections~S1, etc.

\section{General Panel Data Models with Random Effects\label{sec:General-Panel-Data}}
Our motivating example is to investigate the impact of life events
on excessive alcohol consumption $(y_1)$ and mental health $(y_2)$. To match our motivating
example, we describe the following panel data models with two outcomes.
The extension to three or more outcomes is conceptually straightforward,
but can be much more demanding computationally.
\Cref{ss: biv probit with random effects} defines a bivariate probit model which serves as a baseline model
for comparison with the models discussed later.
\Cref{ss: mixed biv model with random effects} then extends the model to accommodate
one continuous outcome and one categorical variable.
Both of these models impose a linear
correlation structure that represents misspecification risk when the
existence of asymmetric or non-linear dependence structures is plausible.
Then, we extend the model such that we still maintain the assumption
of normality for the marginal distributions while introducing non-normal
dependence of error terms using copulas and is given in Section \ref{SS: copulal models}. Lastly, we also define and discuss briefly the Mundlak type specifications which are used for all the models in this section (see
\citet[pg 615-616]{Mundlak:1978,Woolridge:2010}). Note that the models defined in this section can be applied more generally to any variables of interest that can be combination of discrete and continuous variables and they are not restricted to only
the applications in this paper.

\subsection{Bivariate Probit Model with Random Effects}\label{ss: biv probit with random effects}
We first consider the joint distribution
of two binary outcomes given by a bivariate probit model.
Let $y_{1,it}$ and $y_{2,it}$ be the two observed binary outcomes,
for $i=1,...,P$ people and and $t=1,...,T$ time periods.
The bivariate probit model is defined using the following latent variable specification
\begin{align}
y_{1,it}^{*} &=\boldsymbol{x}_{1,it}^\prime\boldsymbol{\beta}_{11}+\overline{\bs x}_{1,i}^{\prime}\boldsymbol{\beta}_{12}+\alpha_{1,i}+\varepsilon_{1,it}
\quad \text{and} \quad
y_{2,it}^{*}=\boldsymbol{x}_{2,it}^{\prime}\boldsymbol{\beta}_{21}+\overline{\bs x}_{2,i}^{\prime}\boldsymbol{\beta}_{22}+\alpha_{2,i}+\varepsilon_{2,it},
\label{eq: biv probit}\\
\intertext{where}
\boldsymbol{\alpha}_{i}&:=\left ( \alpha_{1,i}, \alpha_{2,i}\right )^\prime \sim N \left ( \boldsymbol{0}, \boldsymbol{\Sigma_{\alpha}} \right ) \quad
\text{and} \quad
\boldsymbol{\varepsilon}_{it}=\left ( \varepsilon_{1,it}, \varepsilon_{2,it}\right )^\prime \sim  N \left ( \boldsymbol{0}, \boldsymbol{\Sigma_{\varepsilon}} \right ) \label{eq: alpha and vareps distns}
\intertext{with}
\boldsymbol{\Sigma_{\alpha}} & := \begin{pmatrix} \tau_1^2 & \rho_\alpha \tau_1 \tau_2 \\
\rho_\alpha \tau_1 \tau_2 & \tau_1^2 \end{pmatrix}  \quad \text{and} \quad
\boldsymbol{\Sigma_{\varepsilon}}  := \begin{pmatrix} 1 & \rho_\varepsilon \\
\rho_\varepsilon  & 1 \end{pmatrix}. \label{eq: cov matrices}
\end{align}
In \cref{eq: biv probit}, $\boldsymbol{x}_{j,it},j=1,2,$ are the exogenous variables which may be associated with the two outcomes (including
serious/major life-events) and
\begin{align*}
\overline {\boldsymbol{x}}_{i} &:= \left ( \overline{\boldsymbol{x}}_{1,i}^\prime, \overline{\boldsymbol{x}}_{2,i}^\prime \right)^\prime  = \frac1T\sum_{t=1}^T \boldsymbol{x}_{it}^{(M)}
\end{align*}  is the average over the sample period of the
observations on a subset of the exogenous variables $\boldsymbol{x}_{it}^{(M)}$.
In \cref{eq: biv probit,eq: alpha and vareps distns},
$\boldsymbol{\alpha}_{i}$ is an individual-specific
and time invariant random component comprising potentially correlated
outcome-specific effects, $\boldsymbol{\varepsilon}_{it}$ is the
idiosyncratic disturbance term that varies over time and individuals
and is assumed to be bivariate normally distributed allowing for
contemporaneous correlation but otherwise uncorrelated across individuals
and time and also uncorrelated with $\boldsymbol{\alpha}_{i}$.
 In this model
$y_{1,it}^{*}$ and $y_{2,it}^{*}$ are unobserved. The observed
binary outcomes are defined as
\begin{align}
y_{1,it}:=I\left(y_{1,it}^{*}>0\right) \quad \text{and} \quad
y_{2,it}:=I\left(y_{2,it}^{*}>0\right).   \label{eq: obsns baseline model}
\end{align}
The explanatory variables $\left(\boldsymbol{x}_{it},i=1, \dots, P, t=1, \dots, T\right)$ are also
assumed to be exogenous with respect to the individual random effects
$\boldsymbol{\alpha}_{i}$.

Including
 terms $\overline{\bs x}_{1,i}^{\prime}\boldsymbol{\beta}_{12}$ and
$\overline{\bs x}_{2,i}^{\prime}\boldsymbol{\beta}_{22}$ to \cref{eq: biv probit} is called the \citet{Mundlak:1978} correction
because we can now consider
$\left (\alpha_{1i} + \overline{\bs x}_{1,i}^{\prime}\boldsymbol{\beta}_{12},
\alpha_{2i} + \overline{\bs x}_{2,i}^\prime\boldsymbol{\beta}_{12}
\right )^\prime $ as the $i$th composite random effect which is now potentially correlated
with the exogenous covariates $\bs x_{it}$.

The joint density, conditional to the
vector of individual random effects $\boldsymbol{\alpha}_{i}$,  is
\begin{align}\label{eq: biv probit likel}
p\left(\boldsymbol{y}|\boldsymbol{\theta},\boldsymbol{\alpha}\right)
 &=\prod_{i=1}^{P}\prod_{t=1}^{T}
\Phi_{2}\left(\boldsymbol{\mu}_{it} ,  \boldsymbol{\Sigma}_{\rm probit} \right )
\end{align} 
where $\boldsymbol{\Sigma}_{\rm probit}$ is a $2 \times 2$ covariance matrix  with ones on the diagonal and
$(2y_{1,it}-1)(2y_{2,it}-1)\rho_{\varepsilon}$ on the off diagonal, and
\begin{align}\label{eq: biv probit mu}
\boldsymbol{\mu}_{it} &:=\left (  (2y_{1,it}-1) \left(\boldsymbol{x}_{1,it}^{\prime}\boldsymbol{\beta}_{11}+\boldsymbol{\overline{x}}_{1,i}^{\prime}
\boldsymbol{\beta}_{12}+
\alpha_{1,i}\right),(2y_{2,it}-1)\left(\boldsymbol{x}_{2,it}^{\prime}\boldsymbol{\beta}_{21}+\boldsymbol{\overline{x}}_{2,i}^{\prime}\boldsymbol{\beta}_{22}+
\alpha_{2,i}\right)
\right )^\prime \ .
\end{align}

\subsection{Mixed marginal bivariate Model with Random Effects}\label{ss: mixed biv model with random effects}
We next consider an extension to the baseline model \crefrange{eq: biv probit}{eq: obsns baseline model}, where
we treat one of the outcomes $y_{2,it}=y_{2,it}^{*}$ as an observed continuous variable, with $y_{1,it}$ still discrete as in
\cref{eq: obsns baseline model}. This applies to the mental health variable where continuous values are available.
The joint density conditional on the vector of individual
random effects $\boldsymbol{\alpha}_{i}$ is
\begin{align*}
p\left(\boldsymbol{y}|\boldsymbol{\theta},\boldsymbol{\alpha}\right) & =  \prod_{i=1}^{P}\prod_{t=1}^{T}
\left[\Phi\left(\frac{\mu_{1|2}}{\sigma_{1|2}}\right)\phi\left(y_{2,it}-\left(\boldsymbol{x}_{2,it}^{\prime}
\boldsymbol{\beta}_{21}+\overline{\boldsymbol{x}}_{2,i}^{\prime}\boldsymbol{\beta}_{22}+\alpha_{2,i}\right)\right)\right]^{y_{1,it}}\\
 & \left[\left(1-\Phi\left(\frac{\mu_{1|2}}{\sigma_{1|2}}\right)\right)\phi\left(y_{2,it}-\left(\boldsymbol{x}_{2,it}^{\prime}
 \boldsymbol{\beta}_{21}+\overline{\boldsymbol{x} }_{2,i}^{\prime}\boldsymbol{\beta}_{22}+\alpha_{2,i}\right)\right)\right]^{1-y_{1,it}}
 \intertext{where $\phi(\cdot)$ is the standard normal pdf, $\sigma_{1|2}=\sqrt{1-\rho_\varepsilon^{2}}$ and }
\mu_{1|2} &=\left(\boldsymbol{x}_{1,it}^{\prime}\boldsymbol{\beta}_{11}+\overline{\boldsymbol{x}}_{1,i}^{\prime}\boldsymbol{\beta}_{12}+
\alpha_{1,i}\right)+
\rho_\varepsilon\left(y_{2,it}-\left(\boldsymbol{x}_{2,it}^{\prime}\boldsymbol{\beta}_{21}+
\overline{\boldsymbol{x}}_{2,i}^{\prime} \boldsymbol{\beta}_{22}+\alpha_{2,i}
 \right)\right )
\end{align*}
We assume here and in \Cref{ss: biv probit with random effects}
that the error term $\bs \varepsilon_{it}$ is bivariate normal and consequently impose
a linear correlation structure.
\Cref{SS: copulal models} maintains the assumption of normality
for the marginal distributions of $\bs \varepsilon_{it}$, while introducing non-normal dependence
of the error terms using bivariate copulas.

\subsection{Bivariate Copula Models}  \label{SS: copulal models}

Copula based models provide a flexible approach to multivariate modeling because they can:
 (i)~capture a wide range of non-linear dependence between the marginals beyond simple linear correlation;
  (ii)~allow the marginal distributions to
come from different families of distributions; and, in particular, (iii)~allow the marginal distributions
to be a combination of discrete and continuous distributions as in \cref{ss: mixed biv model with random effects}.
There are many possible parametric copula functions proposed in the
Statistics and Econometric literatures, with the choice of parametric copula determining
the dependence structure of the variables being analysed.
\citet{Trivedi2005} discuss some of the most popular copulas.
A major difference between copula distribution functions is the range of their
dependence structures.
Our article considers the Gaussian, Gumbel and
Clayton copulas, which are three of the most commonly used
bivariate copulas and they are able to capture a wide range of dependence structures. Note that the baseline model is a Gaussian copula. \Cref{sec:Archimedian-and-Elliptical copulas} gives some details of copula models. For a discussion of copula based models with a combination of discrete and continuous marginals,
and their estimation methods see \cite{Pitt2006} and  \cite{Smith2012}. Their method augments the copula model with latent variables which are generated within an MCMC scheme. Note that in our estimation, we do not generate any copula latent variables and we work directly with the conditional density given in Equation \eqref{eq: mixed copula model} below. They also do not consider any individual random effects in their models.

We use the copula framework to obtain a more flexible joint distribution
for  $\bs {\varepsilon}_{it}$, while assuming that its two marginals are normally distributed. Let
$c(\cdot; \bs \theta_{\rm copula}) $ be a bivariate copula density with $\bs{\theta}_{\rm copula}  $
the vector of parameters of the copula. Suppose that $\bs{u}_{it}$ has density $c(\cdot; \bs \theta_{\rm copula}) $ and define
$\varepsilon_{j,it}:= \Phi^{-1} (u_{j,it})$ for  $ j=1,2$, where $\Phi(\cdot)$ the standard normal cdf.
Then the density of $\boldsymbol{\varepsilon}_{it}$ is
\begin{align} \label{eq: biv copula}
c(\bs{u}_{it}; \boldsymbol{\theta}_{\rm copula}) \phi(\varepsilon_{1,it}) \phi( \varepsilon_{2,it})
\end{align}

The joint density of the observations conditional on the vector of individual random
effects for the model in \cref{ss: mixed biv model with random effects} is
\begin{eqnarray} \label{eq: mixed copula model}
p\left(\boldsymbol{y}|\boldsymbol{\theta},\boldsymbol{\alpha}\right) & = & \prod_{i=1}^{P}\prod_{t=1}^{T}\left[\left(1-C_{1|2}\left(u_{1,it}|u_{2,it};\boldsymbol{\theta}_{cop}\right)\right)\phi\left(y_{2,it}-\left(\boldsymbol{x}_{2,it}^{T}\boldsymbol{\beta}_{21}+\overline{\boldsymbol{x}}_{2,i}^{T}\boldsymbol{\beta}_{22}+\alpha_{2,i}\right)\right)\right]^{y_{1,it}}\nonumber \\
 &  & \left[C_{1|2}\left(u_{1,it}|u_{2,it};\boldsymbol{\theta}_{cop}\right)\phi\left(y_{2,it}-\left(\boldsymbol{x}_{2,it}^{T}\boldsymbol{\beta}_{21}+\overline{\boldsymbol{x}}_{2,i}^{T}\boldsymbol{\beta}_{22}+\alpha_{2,i}\right)\right)\right]^{1-y_{1,it}},\label{eq: mixed biv model}
\end{eqnarray}
where $u_{1,it}=\Phi\left(-\left(\boldsymbol{x}_{1,it}^{\prime}\boldsymbol{\beta}_{11}+\overline{\boldsymbol{x}}_{1,i}^{'}\boldsymbol{\beta}_{12}+
\alpha_{1,i}\right)\right)$,
and $u_{2,it}=\Phi\left(y_{2,it}-
\left(\boldsymbol{x} _{2,it}^{'}\boldsymbol{\beta}_{21}+\overline{\boldsymbol{x} }_{2,i}^{'}\boldsymbol{\beta}_{22}+\alpha_{2,i}\right)\right)$.

The conditional distribution function of $U_1$ given $U_2$ in the copula $C\left ( \bs u;\bs{\theta}_{\rm cop}\right)$ is
\begin{align*}
C_{1|2}\left(u_{1}|u_{2};\bs{\theta}_{\rm cop}\right) & =\frac{\partial}{\partial u_{2}}C\left(\bs u ;\bs{\theta}_{\rm cop}\right),\quad
\text{where} \quad
C\left(u_{1},u_{2};\bs{\theta}_{\rm cop}\right)  =\int_{0}^{u_{1}}\int_{0}^{u_{2}}c\left(s_{1},s_{2};\bs{\theta}_{\rm cop}\right){\rm d} s_{1}{\rm d}s_{2}
\end{align*}
and $c\left(\bs u ;\bs \theta_{\rm cop}\right)$ is the density of the copula.

\Cref{sec:Archimedian-and-Elliptical copulas} gives closed form expressions for
the conditional copula distribution functions for the bivariate copulas
used in our article.

The Pearson  correlation coefficient is unsuitable
for  comparing the dependence structures implied
by the different copula models with that of the Gaussian copula,
because it only measures linear dependence.
\Cref{sec:Measures-of-Dependence} discusses Kendall's $\tau$ and the upper and lower tail dependence measures
that we use in the article.

\section{Bayesian Inference and particle Markov chain Monte Carlo samplers} \label{Sec: PMCMC samplers}
This Section discuss efficient Bayesian inference for the random effect
panel data models described in Section \ref{sec:General-Panel-Data}. Our approach is similar to the particle Markov
chain Monte Carlo (PMCMC) approaches of \citet{Andrieu:2010}. However, the
PMCMC approaches in \citet{Andrieu:2010} are derived for state space
models and the random effect panel data models we are interested in this
paper have a different structures, since random effects vary across individual, but they do not change over
time. This requires us to derive the PMCMC approaches for the models
we are interested in from first principles.
The benefit is that the
simple particle structure gives straightforward derivations that make the material more accessible than the current PMCMC literature.
\newline \indent
Let $\bs \theta$ be the parameters in the panel data models described
in Section \ref{sec:General-Panel-Data}. The vector individual random effects is denoted as $\boldsymbol{\alpha}_{1:P}$,
where $P$ is the number of individuals, and the vector of observations for individual $i$ is denoted by $\boldsymbol{y}_i$, for $i = 1, \ldots, P$,
with $\boldsymbol{y}_{1:P}$ denoting all the observations in the sample.
In Bayesian inference, we are interested in sampling from the posterior density
\begin{align} \label{eq: posterior density}
\pi\left(\boldsymbol{\theta},\boldsymbol{\alpha}_{1:P}\right)&:=
p(\boldsymbol{y}_{1:P} |\bs \theta , \bs \alpha_{1:P})p(\bs \alpha_{1:P}|\bs \theta) p(\bs \theta)/Z
\end{align}
where $Z:=p(\boldsymbol{y}_{1:P}) $ is the marginal likelihood.
\newline \indent
The basis of our PMCMC approach is to define a target distribution on an augmented space that includes the parameters and multiple copies of the random effects, which we describe as particles.
This target distribution is used to derive two samplers. The first is the Pseudo Marginal Metropolis-Hasting (PMMH) sampler and the second is the Particle Metropolis within Gibbs (PMwG) sampler.
 The PMMH sampler is an efficient method to sample low dimensional parameters that are highly correlated with the latent states because each
PMMH step generates these parameters without conditioning on the states.
However, in the panel data models defined in \Cref{sec:General-Panel-Data},
the dimension of the parameter space is large so that it is difficult to implement PMMH efficiently. The reason is that it is difficult to obtain good proposals for the parameters that require derivatives of the log-posterior because they cannot be computed exactly
and need to be estimated. The efficiency of PMMH then depends
crucially on how accurately we can estimate the gradient of the log-posterior.
If the error in the estimate of the gradient is too large, then there
will be no advantage in using proposals with derivatives information
over a random walk proposal \citep{Nemeth:2016}. Furthermore, the
random walk proposal has a step size proportional to  $2.56/\sqrt{d}$,
where $d$ is the number of parameters used in the random walk step~\citep{Sherlock:2015}.
This implies that for a large number of parameters, the random walk proposal will move very slowly
and so will be very inefficient.

The PMwG sampler generates the parameters conditioning on the latent random effects
and the parameters of the models can be sampled in separate Gibbs or Metropolis within Gibbs step. Note that by conditioning on the states for the panel data models we are interested in,
we are able to compute the gradient of conditional log-posterior analytically.
Our article uses a Hamiltonian Monte Carlo proposal which requires the gradient of the log posterior to sample the high dimensional parameter $\boldsymbol{\beta}$.  However, this sampler is not very efficient for the parameters that are highly correlated with the latent states. We demonstrate in Section \ref{S: results and discussion} that PMwG sampler performs well for the models that we are interested in this paper.

This section is organised
as follows. \Cref{sec:target distribution} discusses the target distribution. \Cref{sec:Pseudo-Marginal-Metropolis} discusses
the Pseudo Marginal Metropolis-Hastings (PMMH) sampler. \Cref{sec: PMwG} discusses the Particle Gibbs (PG) and Particle Metropolis within Gibbs (PMwG)
samplers. \Cref{sub:Sampling--using Hamiltonian proposal} discusses the Hamiltonian Monte
Carlo proposal. \Cref{sub:TNV} compares our proposed PMwG approach
to some alternative approaches and shows that our method can be much more efficient.

\subsection{Target Distribution}\label{sec:target distribution}

We first describe Algorithm~\ref{alg: IS sampling alg} given below, which constructs a particle approximation to the distribution $\pi\left( \rm d  \bs \alpha_{1:P}|\bs \theta \right)$.
Note that all the models in Section~\ref{sec:General-Panel-Data} have the independence properties
\begin{align} \label{eq:independence likelihood}
p(\bs y |\bs \theta) & = \prod_{i=1}^{P} p(\bs y_i |\bs \theta)
\end{align}
and
\begin{align} \label{eq:independence posterior}
\pi\left(\rm d \bs \alpha_{1:P}|\bs \theta \right)& = \prod_{i=1}^{P}\pi\left(\rm d \bs \alpha_{i}|\bs \theta \right) \\
& = \prod_{i=1}^{P} p \left(\rm d \bs \alpha_{i}|\bs \theta, \mathbf{y}_i \right), \nonumber
\end{align}
where the independence property given in \Cref{eq:independence posterior} will be replicated in our particle appoximation.

Let $\left\{ m_{i}\left(\boldsymbol{\alpha}_{i}|\boldsymbol{\theta},\boldsymbol{y}_{i}\right);i=1,...,P\right\} $ be
a family of proposal densities that we will use to
approximate the corresponding densities $\left\{ \pi \left(\boldsymbol{\alpha}_{i}|\boldsymbol{\theta}\right);i=1,...,P\right\} $.
We define
\begin{align*}
S_{i}^{\boldsymbol{\theta}}& \coloneqq \left\{ \boldsymbol{\alpha}_{i}\in \bs \chi:\pi\left(\boldsymbol{\alpha}_{i}|\bs \theta \right)>0\right\}
\quad
{\rm and}\quad
Q_{i}^{\boldsymbol{\theta}} \coloneqq \left\{ \boldsymbol{\alpha}_{i}\in\bs \chi:m_{i}\left(\boldsymbol{\alpha}_{i}|\boldsymbol{\theta},\boldsymbol{y}_{i}\right)>0\right\} .
\end{align*}

The next assumption ensures that  the proposal
densities $m_{i}\left(\boldsymbol{\alpha}_{i}|\boldsymbol{\theta},\boldsymbol{y}_{i}\right)$
can be used to  approximate approximate the corresponding densities $\left\{ \pi \left(\boldsymbol{\alpha}_{i}|\boldsymbol{\theta}\right);i=1,...,P\right\} $
in Algorithm~\ref{alg: IS sampling alg}.
\begin{assumption}\label{ass: use of IS}
We assume that $S_{i}^{\boldsymbol{\theta}}\subseteq Q_{i}^{\boldsymbol{\theta}}$ for any $\boldsymbol{\theta}\in\boldsymbol{\Theta}$ and $i=1,...,P$
\end{assumption}
Note that Assumption \ref{ass: use of IS} is always satisfied in our implementation
because we use the prior density $p\left(\boldsymbol{\alpha}_{i}|\boldsymbol{\theta}\right)$
as a proposal density and the prior  density is positive everywhere.

The generic Monte Carlo Algorithm~\ref{alg: IS sampling alg} proceeds as follows.
\begin{algorithm}[H]\caption{Monte Carlo Algorithm\label{alg: IS sampling alg}}
\begin{description}
\item For $i=1,...,P$,
\begin{description}
\item [Step (1)]

Sample $\boldsymbol{\alpha}{}_{i}^{j}$ from $m_{i}\left(\boldsymbol{\alpha}_{i}|\boldsymbol{\theta},\boldsymbol{y}_{i}\right)$
for $j=1,...,N$.
\item [Step (2)]
Compute the weights $w_{i}^{j}\coloneqq\frac{p\left(\mathbf{y}_{i}|\boldsymbol{\alpha}{}_{i}^{j},\boldsymbol{\theta}\right)p\left(\boldsymbol{\alpha}{}_{i}^{j}|\boldsymbol{\theta}\right)}{m_{i}\left(\boldsymbol{\alpha}_{i}^{j}|\boldsymbol{\theta},\boldsymbol{y}_{i}\right)},$
for $j=1,...,N$.
\item [Step (3)]
Normalised the weights $\bar{w}_{i}^{j}\coloneqq\frac{w_{i}^{j}}{\sum_{k=1}^{N}w_{i}^{k}}$,
for $j=1,...,N$.
\end{description}
\end{description}
\end{algorithm}
To define the joint distribution of the particles given the parameters, let $\bs \alpha_{1:P}^{1:N}:=
\{ \bs \alpha_1^{1:N}, \dots, \bs \alpha_P^{1:N} \}$.
The joint distribution is
\begin{align} \label{eq:particle dist}
\psi^{\theta}\left(\boldsymbol{\alpha}_{1:P}^{1:N}\right) & \coloneqq\prod_{j=1}^{N}\prod_{i=1}^{P}m_{i}\left(\boldsymbol{\alpha}_{i}^{j}|\boldsymbol{\theta},\boldsymbol{y}_{i}\right).
\end{align}

Under Assumption \ref{ass: use of IS}, Algorithm~\ref{alg: IS sampling alg} yields the approximations to $\pi\left( \rm d  \bs \alpha_{1:P}|\bs \theta \right)$
and the likelihood  $p(\bs y|\bs \theta )$  as
\begin{align*}
\widehat{\pi}_{N}\left(\rm d \bs \alpha_{1:P}|\bs \theta \right)& \coloneqq\prod_{i=1}^{P}\left\{ \sum_{j=1}^{N}\bar{w}_{i}^{j}\delta_{\boldsymbol{\alpha}_{i}^{j}}\left(\rm d\boldsymbol{\alpha}_{i}\right)
\right\}
\end{align*}
and
\begin{align} \label{eq:estimated likelihood}
\widehat{p}_{N}(\bs y |\bs \theta) & \coloneqq\prod_{i=1}^{P}\left\{ \frac{1}{N}\sum_{j=1}^{N}w_{i}^{j}\right\} .
\end{align}

It follows straightforwardly that $\widehat{p}_{N}(\bs y |\bs \theta)$ is an unbiased estimator of the likelihood $p(\bs y |\bs \theta)$.
The proof is given in \Cref{sec:lemma1}.

To obtain particle MCMC schemes to estimate $\pi(\bs \theta, \bs \alpha_{1:P})$,
let $\bs k :=(k_1, \dots, k_P)$, with each $k_i \in \{1, \dots, N\}$, and let
$\bs \alpha_{1:P}^{\bs k}:= \{ \bs \alpha_1^{k_1}, \dots, \bs \alpha_P^{k_P}\}$.  For $N \geq 1$, we define the target density
\begin{align} \label{eq: expanded target density}
\wt {\pi}_N(\bs k, \bs \alpha_{1:P}^{1:N}, \bs \theta)&:= \frac{\pi(\bs \theta, \bs \alpha_{1:P}^{\bs k}) }{N^P} \times
\frac{\psi^\theta (\bs \alpha_{1:P}^{1:N} ) }{ \prod_{i=1}^P m_i (\bs \alpha_i^{k_i} |\bs \theta, \bs y_i) }.
\end{align}
\Cref{sec:lemma1} shows that $N^{-P}\pi\left(\boldsymbol{\theta},\boldsymbol{\alpha}_{1:P}^{\boldsymbol{k}}\right)$
is the marginal probability density of $\widetilde{\pi}_{N}\left(\boldsymbol{k},\boldsymbol{\alpha}_{1:P}^{\boldsymbol{k}},\boldsymbol{\theta}\right)$.
 Using the target density in
\Cref{eq: expanded target density}, the next two sections consider two particle based methods for
carrying Markov chain Monte Carlo in panel data models with an intractable
likelihood.

\subsection{Pseudo Marginal Metropolis Hastings (PMMH) sampler\label{sec:Pseudo-Marginal-Metropolis}}
The PMMH sampler is a Metropolis Hastings update on the extended space with target density defined in Equation \eqref{eq: expanded target density}
and the proposal density for $\bs \theta^\ast,
\bs k^\ast, $ and $
{\bs \alpha}_{1:P}^{\ast 1:N}$, given their current values  $\bs \theta, \bs k$ and $\bs \alpha_{1:P}^{1:N}$, as
\begin{align} \label{eq: extended prop}
 q_N(  \bs k^\ast, {\bs \alpha}_{1:P}^{\ast 1:N}, \bs \theta^\ast | \bs k, \bs \alpha_{1:P}^{1:N}, \bs \theta):= q(\bs \theta^\ast|\bs \theta) \times
\psi^{\theta^\ast} \left( {\bs \alpha}_{1:P}^{\ast 1:N} \right) \times \prod_{i=1}^P {\ov {w}^{\ast}_i}^{k^\ast_i}.
\end{align}
Note that this proposal density first samples $\bs \theta^\ast$ from $q(\bs \theta^\ast | \bs \theta) $.
The ${\bs \alpha}_{1:P}^{\ast 1:N}$  are then sampled from $\psi^{\theta^\ast} (\cdot )$.
Finally, $\bs k^\ast$ is sampled from $\prod_{i=1}^P {\ov {w}^{\ast}_i}^{k^\ast_i}$.

We now consider  the ratio of the extended target \cref{eq: expanded target density} and extended proposal \cref{eq: extended prop}
to obtain the PMMH acceptance probability for this target and proposal. This ratio is
\begin{align} \label{eq: ratio of extended target to prop}
\frac{\wt {\pi}_N\left(\bs k^\ast, {\bs \alpha}_{1:P}^{\ast 1:N}, \bs \theta^\ast\right)}{ q_N\left(\bs k^\ast, {\bs \alpha}_{1:P}^{\ast 1:N}, \bs \theta^\ast | \bs k, \bs \alpha_{1:P}^{1:N}, \bs \theta \right)} & = \frac{\pi(\bs \theta^\ast, \bs \alpha_{1:P}^{\ast \bs k^\ast}) }{N^P} \times
\frac{\psi^{\theta^\ast} \left({\bs \alpha}_{1:P}^{\ast 1:N} \right) }{ \prod_{i=1}^P m_i (\bs \alpha_i^{\ast k^\ast_i} |\bs \theta^\ast, \bs y_i) }\times
\Bigg \{q(\bs \theta^\ast|\bs \theta) \times
\psi^{\theta^\ast} \left( {\bs \alpha}_{1:P}^{\ast 1:N} \right) \times \prod_{i=1}^P {\ov {w}^{\ast}_i}^{k^\ast_i}\Bigg \}^{-1}\notag \\
& = \frac{\wh p_N(\bs y|\bs \theta^\ast )p(\bs \theta^\ast)  }{q(\bs \theta^\ast| \bs \theta)}.
\end{align}
Hence, the acceptance probability is
\begin{align}\label{eq: PMMH acc prob}
\min \Bigg \{ 1, \frac{\wh p_N(\bs y|\bs \theta^\ast )p(\bs \theta^\ast)  }{q(\bs \theta^\ast| \bs \theta)} \frac{q(\bs \theta| \bs \theta^\ast)} {\wh p_N(\bs y|\bs \theta )p(\bs \theta)  }\Bigg \},
\end{align}
which is the acceptance probability of a PMMH scheme based on only estimating $\pi(\theta)$.

The next assumption is needed to ensure
that the PMMH algorithm converges.
\begin{assumption}\label{ass: convergence of pmmh}
 The MH sampler of the target density $\pi\left(\boldsymbol{\theta}\right)$
and proposal density $q\left(\boldsymbol{\theta}^\ast|\boldsymbol{\theta}\right)$
is irreducible and aperiodic.
\end{assumption}

\begin{theorem} \label{thm:convergence of pmmh}

Suppose Assumptions \ref{ass: use of IS} and \ref{ass: convergence of pmmh} hold. Then
 the PMMH algorithm with the expanded target density  \cref{eq: expanded target density} and expanded proposal density
 \cref{eq: extended prop}
 generates a sequence
$\left\{ \boldsymbol{\theta}\left(s\right),\boldsymbol{\alpha}_{1:P}\left(s\right)\right\} $ of iterates
whose marginal distributions $\left\{ \mathcal{\mathcal{L}}_{N}\left\{ \boldsymbol{\theta}\left(s\right),\boldsymbol{\alpha}_{1:P}\left(s\right)\in.\right\} \right\}$
satisfy
\[
||\mathcal{\mathcal{L}}_{N}\left\{ \boldsymbol{\theta}\left(s\right),\boldsymbol{\alpha}_{1:P}\left(s\right)\in\cdot \right\} -\pi\left(\cdot \right)||_{TV}\rightarrow0,\qquad as\;s\rightarrow\infty,
\]
where $|| \cdot||_{TV}$ is the total variation norm.
\end{theorem}

The proof is given in \Cref{sec:lemma1}

\subsection{Particle Metropolis within Gibbs (PMwG) sampling} \label{sec: PMwG}
We use the same extended target distribution \cref{eq: expanded target density} as before in order to sample from $\pi\left(\boldsymbol{\theta},\boldsymbol{\alpha}_{1:P}\right)$. Let $\boldsymbol{\theta}=\left(\boldsymbol{\theta}_{1},...,\boldsymbol{\theta}_{R}\right)$
be a partition of the parameter vector $\boldsymbol{\theta} $ into $R$ components and let
$\boldsymbol{\Theta}=\left(\boldsymbol{\Theta}_{1},...,\boldsymbol{\Theta}_{R}\right)$
be the corresponding partition of the parameter space $\boldsymbol{\Theta} $.
We will use the notation
$
\boldsymbol{\theta}_{-r}\coloneqq \left(\boldsymbol{\theta}_{1},...,\boldsymbol{\theta}_{r-1},\boldsymbol{\theta}_{r+1},...,\boldsymbol{\theta}_{R}\right).
$
The Particle Gibbs (PG) and
Particle Metropolis within Gibbs (PMwG) algorithm involves the following
steps.
\begin{algorithm}[H] \caption{PMwG Algorithm \label{alg: PMwG}}

\begin{description}
\item [Step (1)]
For $r=1,...,R$
\begin{description}
\item [Step (1.1)]
Sample  $\boldsymbol{\theta}_{r}^\ast $ from  the proposal $q_{r}\left(\cdot| \boldsymbol{k},\boldsymbol{\alpha}_{1:P}^{\boldsymbol{k}},\theta_r, \boldsymbol{\theta}_{-r}
    \right)$.
\item [Step (1.2)]
Set $\boldsymbol{\theta}_{r} = \boldsymbol{\theta}_{r}^\ast$ with probability
 \begin{align*}
 \min\left \{1,
 \frac{\tilde{\pi}_{N}\left(\bs{\theta}_{r}^\ast|\bs{k},\boldsymbol{\alpha}_{1:P}^{\boldsymbol{k}},\bs{\theta}_{-r}\right) }
 {\tilde{\pi}_{N}\left(\bs{\theta}_{r}|\bs{k}, \boldsymbol{\alpha}_{1:P}^{\boldsymbol{k}},\bs{\theta}_{-r} \right)}
 \times\frac{q_{r}\left(\boldsymbol{\theta}_{r} |\boldsymbol{k},\boldsymbol{\alpha}_{1:P}^{\boldsymbol{k}},
\theta^\ast_r, \boldsymbol{\theta}_{-r}\right)}
{q_{r}\left(\boldsymbol{\theta^\ast}_{r} |\boldsymbol{k},\boldsymbol{\alpha}_{1:P}^{\boldsymbol{k}},\theta_r, \boldsymbol{\theta}_{-r} \right)} \right \}
\end{align*}
\end{description}
\item [Step (2)] Sample $\boldsymbol{\alpha}_{1:P}^{-\bs  k}\sim\tilde{\pi}_{N}\left(\cdot|\boldsymbol{k},\boldsymbol{\alpha}_{1:P}^{-\bs k},\boldsymbol{\theta}\right)$
by running the conditional importance sampling Algorithm \ref{alg: CIS}.
\item [Step (3)] Sample the index vector $\boldsymbol{k} = \left(k_{1},...,k_{P}\right)$ with probability
given by
\[
\tilde{\pi}_{N} \left(k_{1}=l_{1},...,k_{P}=l_{P}|\boldsymbol{\theta},\boldsymbol{\alpha}_{1:P}^{1:N} \right)=\prod_{i=1}^{P}\overline{w}_{i}^{l_{i}},
\]
where $w_i^j := w_i^j (\bs \theta, \boldsymbol{\alpha}_{1:P}^{1:N} ) $ and
$\overline{w}_{i}^{l} := w_i^l / \left ( \sum_{s=1}^N w_i^s \right ) $.
\end{description}

\end{algorithm}

It is straightforward to implement Steps~(1) and (3). We implement Step~(2)
using the Conditional Monte Carlo Algorithm \ref{alg: CIS} given below.
Note that Step (1) might appear unusual, but it leaves the augmented
target posterior density $\widetilde{\pi}^{N}\left(\boldsymbol{k},\boldsymbol{\alpha}_{1:P}^{1:N},\boldsymbol{\theta}\right)$
invariant.
This is related to collapsed Gibbs sampler, see, for example \citet[section 6.7]{Liu:2001a}.

\subsection*{Conditional Monte Carlo}
The expression
\[
\frac{\psi^{\boldsymbol{\theta}}\left(\boldsymbol{\alpha}_{1:P}^{1:N}\right)}{\prod_{i=1}^{P}m_{i}\left(\alpha_{i}^{k_{i}}\right)}
\]
appearing in the target density \cref{eq: expanded target density}
 is the density
of all the variables that are generated by the Monte Carlo
algorithm conditional on $\left(\bs  \alpha_{1:P}^{\bs k},\boldsymbol{k} \right)$.
This is a key element of the PMwG Algorithm \ref{alg: PMwG}.
This update can be understood as updating $N-1$ Monte Carlo samples
together while keeping one Monte Carlo sample fixed in $\widetilde{\pi}_N\left(\boldsymbol{\alpha}_{1:P}^{1:N}|\bs \theta \right)$.

\begin{algorithm}[H] \caption{Conditional Monte Carlo Algorithm\label{alg: CIS}}
\begin{description}
\item [Step (1)]
Fix $\boldsymbol{\alpha}{}_{1:P}^{1}=\boldsymbol{\alpha}{}_{1:P}^{\boldsymbol{k}}$.
\item [Step (2)]
For $i=1,...P,$

\begin{description}
\item [Step (2.1)]
Sample $\boldsymbol{\alpha}{}_{i}^{j}$ from $m_{i}\left(\boldsymbol{\alpha}_{i}|\boldsymbol{\theta},\boldsymbol{y}_{i}\right)$
for $j=2,...,N$.
\item [Step (2.2)]
Compute the importance weights $w_{i}^{j}=\frac{p\left(\mathbf{y}_{i}|\boldsymbol{\alpha}{}_{i}^{j},\boldsymbol{\theta}\right)p\left(\boldsymbol{\alpha}{}_{i}^{j}|\boldsymbol{\theta}\right)}{m_{i}\left(\boldsymbol{\alpha}_{i}^{j}|\boldsymbol{\theta},\boldsymbol{y}_{i}\right)},$
for $j=1,...,N$.
\item [Step (2.3)]
Normalise the weights $\bar{w}_{i}^{j}=\frac{w_{i}^{j}}{\sum_{k=1}^{N}w_{i}^{k}}$,
for $j=1,...,N$.
\end{description}
\end{description}
\end{algorithm}

To derive convergence results for the PMwG sampler in Algorithm~\ref{alg: PMwG} we require the following assumption.

\begin{assumption} \label{ass: gibbs1}
The Metropolis within Gibbs sampler that is defined by the proposals  $q_r\left(\cdot|\boldsymbol{\theta}_{r},\boldsymbol{\theta}_{-r},\boldsymbol{\alpha}_{1:P}\right)$,
for $r=1,...,R$, and $\pi\left(\boldsymbol{\alpha}_{1:P}|\boldsymbol{\theta}\right)$
is irreducible and aperiodic.
\end{assumption}

Assumption \ref{ass: gibbs1}  is satisfied in our applications
because all the proposals and conditional distributions have strickly positive densities.



\begin{theorem}  [Convergence of the PMwG sampler\label{thm: converg of PMwG}]
For any $N\geq2$, the PMwG update is a transition kernel for the invariant
density $\tilde{\pi}^{N}$ defined in \cref{eq: expanded target density}.
If Assumptions~\ref{ass: use of IS} and \ref{ass: gibbs1} hold then the PMwG sampler generates
a sequence of iterates $\left\{ \boldsymbol{\theta}\left(s\right),\boldsymbol{\alpha}_{1:P}\left(s\right)\right\} $
whose marginal distributions $\left\{ \mathcal{\mathcal{L}}_{N}\left\{ \boldsymbol{\theta}\left(s\right),\boldsymbol{\alpha}_{1:P}\left(s\right)\in.\right\} \right\} $
satisfy
\[
||\mathcal{\mathcal{L}}_{N}\left\{ ( \boldsymbol{\theta}\left(s\right),\boldsymbol{\alpha}_{1:P}\left(s\right)) \in\cdot \right\}
-\Pi\left \{ ( \boldsymbol{\theta},\boldsymbol{\alpha}_{1:P})\in \cdot \right\}
||_{\rm TV}\rightarrow0,\qquad as\;s\rightarrow\infty.
\]
\end{theorem}

The proof is given in \Cref{sec:lemma1}.

\subsection{Sampling the high-Dimensional parameter vector  $\bs \beta$ using a Hamiltonian Proposal\label{sub:Sampling--using Hamiltonian proposal}}

This section discusses  the Hamiltonian Monte Carlo (HMC) proposal
to sample the high dimensional parameter vector $\boldsymbol{\beta}$ from the conditional
posterior density $p\left(\boldsymbol{\beta}|\boldsymbol{\theta}_{-\boldsymbol{\beta}},\boldsymbol{y},\boldsymbol{\alpha}_{1:P}^{\boldsymbol{k}},\boldsymbol{k}\right)$.
It can be used to generate distant proposals for the Particle Metropolis
within the Gibbs algorithm to avoid the slow exploration behaviour
that results from simple random walk proposals.

Suppose we want to sample from a $d$-dimensional distribution with
pdf proportional to $\exp\left(\mathcal{L\left(\boldsymbol{\beta}\right)}\right)$,
where $\mathcal{L}\left(\boldsymbol{\beta}\right)=\log p\left(\boldsymbol{\beta}|\boldsymbol{\theta}_{-\boldsymbol{\beta}},\boldsymbol{y},\boldsymbol{\alpha}_{1:P}^{\boldsymbol{k}},\boldsymbol{k}\right)$
is the logarithm of the conditional posterior density of $\boldsymbol{\beta}$
(up to a normalising constant). In Hamiltonian Monte Carlo \citep{Neal:2011},
we augment an auxiliary momentum vector $\bs r$  having the same dimension as the parameter vector $\bs \beta$
 with the density $p\left(\boldsymbol{r}\right)=N\left(\boldsymbol{r}|0,\bs M\right)$,
where $\bs M$ is a mass matrix of the momentum and often set to the identity matrix.
We define the joint conditional density of $(\bs \beta , \bs r ) $ as
\begin{align}
p\left(\boldsymbol{\beta},\boldsymbol{r}|\boldsymbol{\theta}_{-\boldsymbol{\beta}},\boldsymbol{y},\boldsymbol{\alpha}_{1:P}^{\boldsymbol{k}},\boldsymbol{k}\right)\propto\exp\left(-H\left(\boldsymbol{\beta},\boldsymbol{r}\right)\right)\label{eq:jointHamiltonian-1}
\end{align}
where
\begin{align}\label{eq: hamiltonian}
H\left(\boldsymbol{\beta},\boldsymbol{r}\right):=-\mathcal{L\left(\boldsymbol{\beta}\right)} + \frac12 \boldsymbol{r}^\prime M^{-1} \boldsymbol{r}
\end{align}
is called the Hamiltonian.

In an idealized HMC step, the parameters $\boldsymbol{\beta}$
and the momentum variables $\boldsymbol{r}$  move continuously
according to the differential equations
\begin{align}
\frac{d\boldsymbol{\beta}}{\rm d t} & =\frac{\partial H}{\partial\boldsymbol{r}} =\bs M^{-1}\boldsymbol{r}\\
\frac{d\boldsymbol{r}}{\rm d t }&=-\frac{\partial H}{\partial\boldsymbol{\beta}} =\nabla_{\boldsymbol{\beta}}\mathcal{L\left(\boldsymbol{\beta}\right)},
\end{align}
where $\nabla_{\boldsymbol{\beta}}$ denotes the gradient with respect
to $\boldsymbol{\beta}$. In a practical implementation, the continuous time HMC dynamics
need to be approximated by discretizing  time, using a small step size
$\epsilon$. We can simulate the evolution over time of $\left(\boldsymbol{\beta},\boldsymbol{r}\right)$
via the \dlq leapfrog\drq{} integrator, where one step of the leapfrog update is
\begin{eqnarray*}
\boldsymbol{r}\left( t +\frac{\epsilon}{2}\right) & = & \boldsymbol{r}\left( t \right)+\epsilon\nabla_{\boldsymbol{\beta}}\mathcal{L}\left(\boldsymbol{\beta}\left(t \right)\right)/2\\
\boldsymbol{\beta}\left( t +\epsilon\right) & = & \boldsymbol{\beta}\left(t \right)+\epsilon \bs M^{-1}\boldsymbol{r}\left(t +\frac{\epsilon}{2}\right)\\
\boldsymbol{r}\left(t +\epsilon\right) & = & \boldsymbol{r}\left( t +\epsilon/2\right)+\epsilon\nabla_{\boldsymbol{\beta}}\mathcal{L}\left(\boldsymbol{\beta}\left( t +\epsilon\right)\right)/2
\end{eqnarray*}

Each leapfrog step is time reversible by negating the step size
$\epsilon$. The leapfrog integrator provides a mapping $\left(\boldsymbol{\beta}^\ast,\boldsymbol{r}^\ast\right)\rightarrow\left(\boldsymbol{\beta},\boldsymbol{r}\right)$
that is both time-reversible and volume preserving \citep{Neal:2011}.
 It follows that the Metropolis-Hastings algorithm with acceptance probability \[\min\left(1,\frac{\exp\left(\mathcal{L\left(\boldsymbol{\beta}\right)}-\frac{1}{2}\boldsymbol{r}^\prime \bs M^{-1}\boldsymbol{r}\right)}
 {\exp\left(\mathcal{L\left(\bs \beta^\ast\right)}-\frac{1}{2}{\boldsymbol r^\ast}^\prime \bs M^{-1}\boldsymbol{r^\ast}\right)}\right)\]
produces an ergodic, time reversible Markov chain that satisfies detailed
balance and has stationary density $p\left(\boldsymbol{\beta}|\boldsymbol{\theta}_{-\boldsymbol{\beta}},\boldsymbol{y},\boldsymbol{\alpha}_{1:P}^{\boldsymbol{k}},\boldsymbol{k}\right)$
\citep{Liu:2001a,Neal:1996}.  \Cref{alg:Hamiltonian-Monte-Carlo} summarizes a single iterate of the  Hamiltonian Monte Carlo
method.

\begin{algorithm}[h!]
\caption{Hamiltonian Monte Carlo\label{alg:Hamiltonian-Monte-Carlo}}

Given $\boldsymbol{\beta}^{*}$, $\epsilon$, $L$, where $L$ is
the number of Leapfrog updates.

Sample $\boldsymbol{r}^{*}\sim N\left(0,\bs M\right)$.

For $i=1$ to $L$

Set $\left(\boldsymbol{\beta},\boldsymbol{r}\right)\leftarrow {\rm Leapfrog}\left(\boldsymbol{\beta}^{*},\boldsymbol{r}^{*},\epsilon\right)$

end for

With probability $\alpha=\min\left(1,\frac{\exp\left(\mathcal{L\left(\boldsymbol{\beta}\right)}-\frac{1}{2}\boldsymbol{r}^{'}\bs M^{-1}\boldsymbol{r}\right)}{\exp\left(\mathcal{L\left(\boldsymbol{\beta}^{*}\right)}-\frac{1}{2}\boldsymbol{r}^{*'}\bs M^{-1}\boldsymbol{r}^{*}\right)}\right)$,
then set $\boldsymbol{\beta}^{*}=\boldsymbol{\beta}$, $\boldsymbol{r}^{*}=-\boldsymbol{r}$.
\end{algorithm}

The performance of HMC depends strongly on choosing suitable values
for $\bs M$, $\epsilon$, and $L$. We set $\bs M=\widehat{\bs \Sigma}^{-1}$,
where $\widehat{\bs \Sigma}$ is an estimate of the posterior covariance matrix after
some preliminary pilot runs of the HMC algorithm. The step size $\epsilon$
determines how well the leapfrog integration can approximate the Hamiltonian
dynamics. If we set $\epsilon$ too large, then the simulation error
is large yielding a low acceptance rate. However, if we set $\epsilon$
too small, then the computational burden is too high to obtain distant
proposals. Similalry, if we set $L$ too small, the proposal
will be close to the current value of the parameters, resulting in undesirable
random walk behaviour and slow mixing. If $L$ is too large, HMC will
generate trajectories that retrace their steps. Our article
uses the No-U-Turn sampler (NUTS) with the dual averaging algorithm developed
by \citet{Hoffman:2014} and \citet{Nesterov:2009}, respectively,
that still leaves the target density invariant and satisfies time
reversibility to adaptively select $L$ and $\epsilon$, respectively.

\Crefrange{sec:Sampling-Scheme-for bivariate probit}{sec:Sampling-Scheme-for mixed gumbel}
give the derivatives required by the Hamiltonian dynamics for the panel
data models given in \cref{sec:General-Panel-Data}.

\subsection{Comparing the performance of the  PMwG with some other approaches \label{sub:TNV}}

We now specialize the PMwG sampling scheme described in \Cref{alg: PMwG},
to the bivariate probit model with random effects to obtain~\Cref{alg: pmwg for biv probit} below;
see \Cref{sec:Sampling-Scheme-for bivariate probit} for
 more details.

Let $\boldsymbol{\theta}=\left(\boldsymbol{\beta}_{1},\boldsymbol{\beta}_{2},\rho,\bs \Sigma_{\alpha}\right)$
be the set of unknown parameters of interest. We use the following
prior distributions: $\rho\sim U\left(-1,1\right)$, $p\left(\bs \Sigma_{\alpha}^{-1}\right)\sim {\rm Wishart}\left(v_{0},\bs R_{0}\right)$,
where $v_{0}=6, \bs R_{0}=400I_{2}$, and  the prior distribution for the parameters of
the covariates is $N\left(0,100\bs I_{d}\right)$.
All the priors are uninformative.
\begin{algorithm}[H]\caption{PMwG sampling scheme for bivariate probit \label{alg: pmwg for biv probit}}
\begin{enumerate}
\item Generate $\bs \Sigma_{\alpha}|\boldsymbol{k}^{*},\boldsymbol{\alpha}_{1:P}^{\boldsymbol{k}^{*}},\boldsymbol{\theta}_{-\Sigma_{\alpha}}^{*},\boldsymbol{y}$
from a Wishart $W\left(v_{1},R_{1}\right)$ distribution, where $v_{1}=v_{0}+P$
and $\bs R_{1}=\left[\bs R_{0}^{-1}+\sum_{i=1}^{P}\boldsymbol{\alpha}_{i}\boldsymbol{\alpha}_{i}^{'}\right]^{-1}.$
\item Generate $\rho|\boldsymbol{k}^{*},\boldsymbol{\alpha}_{1:P}^{\boldsymbol{k}^{*}},\boldsymbol{\theta}_{-\rho}^{*},\boldsymbol{y}$
using the adaptive random walk proposal described below.
\item Generate $\left(\bs \beta_{1},\bs \beta_{2}\right)|\boldsymbol{k}^{*},\boldsymbol{\alpha}_{1:P}^{\boldsymbol{k}^{*}},\boldsymbol{\theta}_{-\left(\bs \beta_{1},\bs \beta_{2}\right)}^{*},\boldsymbol{y}$
using PMwG with Hamiltonian proposal described in the Section \ref{sub:Sampling--using Hamiltonian proposal}.
\item Sample from $\boldsymbol{\alpha}_{1:P}^{-k_{1}:-k_{P}}\sim\tilde{\pi}^{N}\left(\cdot |\boldsymbol{k}^{*},\boldsymbol{\alpha}_{1:P}^{\boldsymbol{k}^{*}},\boldsymbol{\theta}^{*}\right)$
using Conditional Monte Carlo method.
\item Sample $\left(k_{1},...,k_{P}\right)$ with probability given by $\Pr\left(k_{1}=l_{1},...,k_{P}=l_{P}|\boldsymbol{\theta},
    \boldsymbol{\alpha}_{1:P}^{-\boldsymbol{k}},\boldsymbol{\alpha}_{1:P}^{\boldsymbol{k}^{*}},\boldsymbol{y}\right)
    =\prod_{i=1}^{P}\overline{w}_{i}^{l_{i}}$.
\end{enumerate}
\end{algorithm}
In Step 2 we transform $\rho$ to $\rho_{\rm un} ={\rm tanh}^{-1} (\rho)$ so that $\rho_{\rm un} $  is unconstrained
and then use the adaptive random walk method of
 \citet{Garthwaite:2015}
that automatically scales univariate Gaussian random walk proposals
to ensure that the acceptance rate is around 0.3. The MH acceptance probability is
\[
1\land\frac{\tilde{\pi}^{N}\left(\rho|\boldsymbol{k}^{*},\boldsymbol{\alpha}_{1:P}^{\boldsymbol{k}^{*}},\boldsymbol{\theta}_{-\rho}^{*}\right)}{\tilde{\pi}^{N}\left(\rho^{*}|\boldsymbol{k}^{*},\boldsymbol{\alpha}_{1:P}^{\boldsymbol{k}^{*}},\boldsymbol{\theta}_{-\rho}^{*}\right)}\frac{\left|1-\left(\rho\right)^{2}\right|}{\left|1-\left(\rho^{*}\right)^{2}\right|}.
\]

We can alternatively replace steps 4 and 5, and sample the latent random effects using the Metropolis-Hastings
algorithm, which is denoted by MCMC-MH and is given as
\begin{itemize}
\item $4^{*})$ Sample $\boldsymbol{\alpha}_{i}\sim p\left(\boldsymbol{\alpha}_{i}|\boldsymbol{\theta}\right)=N\left(0,\bs \Sigma_{\alpha}\right)$
for $i=1,\dots, P.$
\item $5^{*})$ Accept $\boldsymbol{\alpha}_{i}$ with acceptance probability given
by
\[
\alpha\left(\boldsymbol{\alpha}_{i},\boldsymbol{\alpha}_{i}^{*}|\boldsymbol{\theta}\right)=1\land\frac{\prod_{t=1}^{T}p\left(y_{it}|\boldsymbol{\theta},\boldsymbol{\alpha}_{i}\right)}{\prod_{t=1}^{T}p\left(y_{it}|\boldsymbol{\theta},\boldsymbol{\alpha}_{i}^{*}\right)}.
\]

\end{itemize}
To improve the mixing of the MCMC-MH algorithm, we can
run steps $4^{*}$ and $5^{*}$ of MCMC-MH method above for a number of iterations, say $10$, $20$,
or $50$ iterations, for each individual random effect. \Cref{sec:Data-Augmentation-Bivariate probit} describes an alternative
Gibbs sampling scheme with data augmentation for
the bivariate probit with random effects model.

We conducted a simulation study to compare three different
approaches to estimation: PG, data augmentation, and MCMC-MH, using
the base case bivariate probit model with random effects as the data
generating process. To define our measure of the inefficiency of different
sampling schemes that takes computing time into account, we first
define the Integrated Autocorrelation Time $\left(IACT_{\boldsymbol{\theta}}\right)$.
For a univariate parameter $\theta$, the IACT is estimated by
\[
\widehat{IACT}\left(\theta_{1:M}\right)\coloneqq1+2\sum_{t=1}^{L}\widehat{\rho}_{t}\left(\theta_{1:M}\right),
\]
where $\widehat{\rho}_{t}\left(\theta_{1:M}\right)$ denotes the empirical
autocorrelation at lag $t$ of $\theta_{1:M}$ (after discarding the burnin period
iterates).

A low value of the IACT estimate suggests that the chain
mixes well. Here, $L$ is chosen as the first index for which the
empirical autocorrelation satisfies $\left|\widehat{\rho}_{t}\left(\theta_{1:M}\right)\right|<2/\sqrt{M}$,
i.e. when the empirical autocorrelation coefficient is statistically
insignificant. Our measure of inefficiency of the sampling scheme
is the time normalised variance
\[
TNV\coloneqq IACT_{mean}\times CT,
\]
where $CT$ is the computing time and $IACT_{mean}$ is the average  of the
IACT's over all parameters.

For this simulation study, we generated a number of datasets with $P=1000$ people
and $T=4$ time periods.
The covariates are generated as $x_{1,it},...,x_{10,it}\sim U\left(0,1\right)$,
and the parameters are set as
\begin{align*}\boldsymbol{\beta}_{1} &=\left(-1.5,0.1,-0.2,0.2,-0.2,0.1,-0.2,0.1,-0.1,-0.2,0.2\right)^{\prime}\\
\boldsymbol{\beta}_{2} &=\left(-2.5,0.1,0.2,-0.2,0.2,0.12,0.2,-0.2,0.12,-0.12,0.12\right)^{\prime}
\end{align*}
with $\tau_{1}^{2}=2.5$, $\tau_{2}^{2}=1$, and, $\rho_{\epsilon}=\rho_{\alpha}=0.5$.
In the simulation study, the total number of MCMC iterations was 11000,
with the first 1000 discarded as burnin. The number of
importance samples in the PMwG method was set as $100$.

\Cref{tab:Comparison-of-Different biv probit2-1}
summarises the estimation results and show that the PMwG sampler performs best. Tables \ref{tab:Comparison-of-Different biv probit1} and \ref{tab:Comparison-of-Different biv probit2} in Appendix \ref{tab:Comparison-of-Different bivariate probit} show the inefficiency factors (IACT) for each parameters in the bivariate probit model.
In terms of TNV, PMwG is more than twice as good as the data
augmentation approach and is also $7.51$, $8.35$, $14.22$, and
$25.40$ times better than MCMC-MH with $1$, $10$, $20$, and $50$
iterations, respectively. This gain is mostly due to the faster computing
time (CT) of the PG over MCMC-MH method. Note that with PG, the computation
of importance weights in the Conditional Monte Carlo to sample
each individual latent random effects can easily be paralellised.
On the other hand, the MCMC-MH approach is a sequential method that
may not easily be parallelised. The full Gibbs sampler with the data augmentation
approach may not be available for all of the models one might want
to consider. The high dimensional parameter vector $\bs \beta$ is sampled much
more efficiently using Hamiltonian proposals compared to the data
augmentation approach, which confirms the usefulness of Hamiltonian
proposals for such high dimensional parameters.

We also ran a second
simulation study for the mixed discrete-linear Gaussian regression.
\Cref{sec:Simulation-Mixed-Discrete} reports the results.

\begin{table}[H]
\caption{TNV comparison of different sampling schemes (PG, data augmentation,
MCMC-MH) for the bivariate probit regression simulation with random
effects with $P=1000$ and $T=4$ \label{tab:Comparison-of-Different biv probit2-1}}

\centering{}%
\begin{tabular}{ccccccc}
\hline
& PG & Data Aug. & MH1 & MH10 & MH20 & MH50\tabularnewline
\hline
Time & $0.62$ & $0.13$ & $0.48$ & $3.21$ &  $6.56$ & $15.13$\tabularnewline
$IACT_{mean}$ & $4.42$ & $44.23$ & $42.88$ & $7.13$ & $5.94$ & $4.60$\tabularnewline
TNV & $2.74$ & $5.75$ & $20.58$ & $22.89$ & $38.97$ & $69.60$\tabularnewline
Rel. TNV & $1$ & $2.09$ & $7.51$ & $8.35$ & $14.22$ & $25.40$\tabularnewline
\hline
\end{tabular}
\end{table}

\section{The Data and their characteristics}\label{S: data and characterization}

\subsection{Sample and Variable Definitions}\label{SS: sample and variable defs}

To estimate the impact of life-shock events on the two outcomes, alcohol
consumption, especially the propensity to binge drink, and the level
of mental health, we use data from Release 14 of the Household, Income,
and Labour Dynamics in Australia (HILDA) survey. HILDA is a nationally
representative longitudinal survey, which commenced in Australia in
2001, with a survey of 13969 persons in 7682 households, and  is
conducted annually. Each year, all household members aged 15 years
or older were interviewed and considered as part of the sample. Information
is collected on education, income, health, life satisfaction, family
formation, labour force dynamics, employment conditions, and other
economic and subjective well-being. In our paper, the analysis is at
the level of the individual,  where we include those who are aged 15
years or older who have non-missing information on our two outcome
variables: life shock variables, and other independent variables.
We use balanced samples. \footnote{The data used in this paper was extracted using the Add-On package
PanelWhiz for Stata. PanelWhiz (http://www.PanelWhiz.eu) was written
by Dr. John P. Haisken-DeNew (john@PanelWhiz.eu). See \citet{Hahn:2013}
and \citet{Haisken-DeNew:2010} for details. }

Following \citet{Frijters:2014}, the data on mental health status
used in this paper is generated from nine questions included in the
Short-Form General Health Survey (SF-36), which is available in all
the waves.

We construct a mental health score by taking the mean of the responses by the individual
 and then standardise so that the index has a mean zero
and standard deviation one. Lower scores indicate better mental
health status. \citet{Butterworth:2004} provide evidence that
the SF-36 data collected in the HILDA survey are valid and can be
used as a general measure of physical and mental health status. We
then categorise someone with good mental health if their score is below
0, and someone with poor mental health if their score is above 0.

The data on alcohol consumption used in this paper is generated from
two questions in the HILDA survey. Subjects are asked to respond
to the question: Do you drink alcohol? The second question
 we considered is related to the problem of binge drinking and is only available
in waves  7, 9, 11, and 13. Respondents identified as drinkers from
the first question are asked: how often do you have 5 or more (female)
or 7 or more (male) standard drinks on any one occasion?

Similarly to \citet{Srivastava:2010}, we define the composite binary variable
FREQUENT\_BINGE as FREQUENT\_BINGE=1 if a male respondent drinks excessive
alcohol more than 1 day per week or a female respondent drinks alcohol
more than 2 or 3 times a month, and is zero otherwise.

The life-shock indicators are generated from responses in a section
of HILDA's self-completion questionaire. Respondents are told \lq We
now would like you to think about major events that have happened
in your life over the past 12 months\rq{} and are asked whether any of
the following apply to them: (1) Separated from spouse or long-term
partner, (2) Serious personal injury or illness to self, (3) Death
of spouse/child, (4) Got back together with spouse or long-term partner
after a separation, (5) Death of a close friend, (6) Victim of property
crime (e.g. theft, housebreaking), (7) Got married, (8) Promoted at
work, (9) Major improvement in financial situation, (10) Major worsening
in financial situation, (11) Changed residence, (12) Partner or I
gave birth to a child.

\subsection*{Other Variables}

Other control variables included are marital status (married, single/widow/divorce)
and the highest educational qualification attained (degree, diploma/certificate,
high school and no qual). single/widow/divorce is the excluded category
for marital status. Similarly, high school and no qual (no academic
qualification) is excluded for the educational variable. We also include
the total number of children below age 18 living in the household,
age, and the logarithm of annualised household income. The Mundlak correction contains $\overline{age}$, $\overline{age}^{2}$,
$\overline{\log\left(income\right)}$, $\overline{num.child}$.

\section{Results and Discussion }\label{S: results and discussion}

\subsection{Estimation Results}\label{SS: estimation results}

The PMwG sampling scheme was used to estimate the various models
defined in \cref{sec:General-Panel-Data}. For
each panel data model, we used 11000 MCMC samples of which the first 1000
were discarded as burnin. After convergence, $M=10,000$ iterates
$\left\{ \boldsymbol{\theta}^{\left(m\right)}\right\} $ were collected
from which we estimated the posterior means of the parameters as well
as their 95\% credible intervals. The Bayesian methodology provides information
on the entire posterior distribution not only for the parameters of
the models, but also other parameters of interest especially the partial
effects. We say that the variable of interest is significant if its
95\% posterior probability intervals does not cover zero.

Models for men and women were analysed separately throughout.
\Cref{tab:Estimation-Results-for Male dependence parameters,tab:Estimation-Results-for feMale dependence parameters}
in \cref{sec:Empirical-Results} show the estimates for various
model specifications for the dependence parameters. The Clayton and Gumbel
copula specifications are used for the contemporaneous error terms. The
overall pattern of dependence is similar for males and females. The
dependence in the individual effects and the error terms of the two outcomes
are weak for both male and female models as measured by the Kendall
tau, which we denote as $\kappa_{\tau}$. The lower tail dependence based on the Clayton copula
is very close to zero for both males and females. This indicates that
there is little relationship between the unobservables who are in very good
mental health and no excessive alcohol consumption in our data. Furthermore,
the upper tail dependence based on the Gumbel copula is also very
close to zero for both males and females. This also suggests that
there is a weak relationship between having very poor mental health
and excessive alcohol consumption after conditioning on the covariates.
The estimates of the dependence parameters from the bivariate probit
model are similar to those from the Gaussian model specification and
are consistent with the expected positive correlation although only
one of the correlations is significant. The estimate of $\tau_{2}^{2}$ is much bigger for the bivariate probit
model than for the Gaussian model because some information is lost in going to binary variables from continuous variables.
Tables \ref{tab:Estimation-Results-for male binge } to \ref{tab:Estimation-Results-for female mental health}
in
\cref{sec:Empirical-Results} give the estimates of the parameters
of the main covariates. We do not report the estimates associated
with the Mundlak corrections for conciseness. The patterns of the estimates for the covariates in the binge drinking equation $(y_1)$ from the bivariate probit model and the Gaussian model specification
are relatively similar across males and females. However, the estimates of the covariates in the mental health equation $(y_2)$ are slightly different and the Gaussian model has tighter posterior probability intervals. All the copula models gave similar results. Furthermore, it
can be seen that all the parameters $\left(\boldsymbol{\beta}_{1},\boldsymbol{\beta}_{2}\right)$
are estimated efficiently for all the models. For males,
the mean IACT for $\boldsymbol{\beta}_{1}$ is $2.55$, $2.55$, $2.59$,
and $2.60$ and the mean IACT for $\boldsymbol{\beta}_{2}$ is $1.41$,
$2.43$, $1.30$, and $1.51$ for the mixed Gaussian, bivariate probit,
mixed Clayton, and mixed Gumbel, respectively. For females, the mean
IACT for $\boldsymbol{\beta}_{1}$ is $2.35$, $2.41$, $2.41$, and
$2.27$ and the mean IACT for $\boldsymbol{\beta}_{2}$ is $1.58$,
$2.09$, $1.41$, and $1.18$ for mixed Gaussian, bivariate probit,
mixed Clayton, and mixed Gumbel, respectively. This confirms the usefulness
of Hamiltonian Monte Carlo proposals for a high dimensional parameter
$\boldsymbol{\beta}$.

Our primary interest is in the impact of the shocks on the joint outcomes
of binge drinking and poor mental health. For these we compute average
partial effects as described in the next section.

\subsection{Average Partial Effects}\label{SS: average partial effects}
We use an Average Partial Effect (APE) to study how a life event such
\lq victim of a crime\rq{} impacts on the association between the
joint outcomes of binge drinking and low mental health. Let $A_{it}$ denote the
event that person $i$ at time $t$ both binge drinks and has poor mental health.

We define the APE for a particular life event LE as
\begin{align} \label{eq: ape def}
\rm {APE}_{ \rm LE}&:= \frac{1}{PT}\sum_{i=1}^P\sum_{t=1}^T \int \left [
\Pr\left (A_{it}|  \bs x_{it}^{(1)}, \ov {\bs x_i} , \bs \theta, \bs y \right )
-  \Pr\left (A_{it}|  \bs x_{it}^{(0)}, \ov {\bs x_i} , \bs \theta, \bs y \right )
\right ] \pi(\bs \theta) { \rm d } \bs \theta,
\end{align}
where $\pi(\bs \theta)$ is the posterior density of $\bs \theta$ and
the superscript $(1)$ in $\bs x_{it}^{(1)}$ means that the life event of interest is set to 1, with a similar
interpretation for $\bs x_{it}^{(0)}$. That is, $\rm {APE}_ { \rm LE}$ is the average over all people and time
periods of the posterior probability of both binge drinking and poor mental health given the data.

Due to the similarity of the results across copula models, we only present the results for the Gaussian copula and
the bivariate probit models.
Given the draws $\{ \bs \theta ^{(m)}, m=1, \dots, M\}$ from the posterior $\bs \theta$, the estimate of the
${\rm APE}_{\rm LE}$ for the bivariate probit model is
\begin{align*}
\widehat{APE}_{ \rm LE} & =  \frac{1}{M}\sum_{m=1}^{M}\frac{1}{PT}\sum_{i=1}^{P}\sum_{t=1}^{T}
\left [
\Phi_{2}\left(\left (\bs \zeta_{it}^{(1)}\right )^{(m)} ; \bs \Sigma^{(m)} + \bs \Sigma_\alpha ^{(m)} \right )
- \Phi_{2}\left(\left (\bs \zeta_{it}^{(0)}\right )^{(m)} ; \bs \Sigma^{(m)} + \bs \Sigma_\alpha ^{(m)} \right )
\right ],
\intertext{where}
\bs \zeta_{it} & =\boldsymbol{\zeta}_{it}(\bs x_{it}, \ov{ \bs x}_i, \bs \theta):=\left(\boldsymbol{x}_{1,it}^{T}\boldsymbol{\beta}_{1,1}+
\boldsymbol{\overline{x}}_{1,i}\boldsymbol{\beta}_{1,2},\boldsymbol{x}_{2,it}^{T}\boldsymbol{\beta}_{2,1}+
\boldsymbol{\overline{x}}_{2,i}\boldsymbol{\beta}_{2,2}\right)^{T}, \left (\bs \zeta_{it}^{(1)}\right )^{(m)} = \bs \zeta_{it}(\bs x_{it}^{(1)}, \bs {\ov x_i}, \bs \theta^{(m)},
\end{align*}
 with
$\left (\bs \zeta_{it}^{(0)}\right )^{(m)}$ defined similarly, $\Sigma^{(m)} = \bs \Sigma ( \bs \theta^{(m)}    ) $
and $\bs \Sigma_\alpha^{(m)} = \bs \Sigma_\alpha  ( \bs \theta^{(m)}    ) $.

\Cref{tab:Average-Partial-Effects for male,tab:Average-Partial-Effects for female}
summarize the estimates of the APEs for major life events shocks for
the probability of binge drinking and low mental health score for
both males and females for the bivariate probit and Gaussian copula models.
In this case, the sign of the APE provides a clear qualitative interpretation,
with a significant positive sign implying a significant increase in
the joint probability of binge drinking and having low mental health
score, and vice versa. All the IACTs of the APEs, for both males and females,
for all life events are very close to $1$, showing that the APEs are
estimated efficiently.

Although these APE effects are small, it is necessary to compare
them with the unconditional joint probability of binge drinking and poor
mental health which are also small for both males and females. For
example, the effect of a personal injury for males is a bit over 2
percentage points for both models but when expressed as a percentage
of the unconditional probability it is 37\% according to the probit estimates
and 38\% for the Gaussian estimates. The death of a spouse/child also
has large relative effects for males but these are not significant.
In fact, the only significant APE for the joint probability is for
personal injury despite several of the estimates in the marginal models
reported in  \cref{tab:Estimation-Results-for male binge ,tab:Estimation-Results-for male mental health} being significant.

In general, the results across the  probit and Gaussian models
are very similar for both males and females. However, unlike the results
for males, several shocks have large and significant effects for females.
Being separated from their spouse, changing residence, a worsening financial
situation, and having a promotion at work are the shocks that are
all significant for females. Each of these increases the joint probability
of binge drinking and poor mental health and, using the Gaussian results,
have impacts relative to the unconditional probability ranging from
16\% for ,a change in residence to 58\% for a worsening in the financial position.
Finally, the \lq gave birth\rq{}  shock has a significant negative association
which reduces the joint probability of binge drinking and poor mental
health.

\begin{table}[H]
\caption{Average Partial Effects on the probability of binge drinking and mental
health for major life events variables for male. Symbol $^{*}$ denotes statistical significance \label{tab:Average-Partial-Effects for male}}

\centering{}%
\begin{tabular}{ccccc}
\hline
Variables & Probit & IACT & Gaussian & IACT\tabularnewline
\hline
\hline
gave birth & $\underset{\left(-0.02,0.00\right)}{-0.01}$ & 1.17 & $\underset{\left(-0.03,0.00\right)}{-0.01}$ & 1.22\tabularnewline
death of a friend & $\underset{\left(-0.00,0.01\right)}{0.00}$ & 1.28 & $\underset{\left(-0.01,0.01\right)}{0.00}$ & 1.38\tabularnewline
death of a spouse/child & $\underset{\left(-0.01,0.08\right)}{0.03}$ & 1.05 & $\underset{\left(-0.02,0.08\right)}{0.02}$ & 1.12\tabularnewline
personal injury & $\underset{\left(0.01,0.04\right)^{*}}{0.02}$ & 1.19 & $\underset{\left(0.01,0.04\right)^{*}}{0.02}$ & 1.45\tabularnewline
getting married & $\underset{\left(-0.02,0.01\right)}{-0.01}$ & 1.15 & $\underset{\left(-0.03,0.01\right)}{-0.01}$ & 1.14\tabularnewline
changed residence & $\underset{\left(-0.01,0.01\right)}{-0.00}$ & 1.24 & $\underset{\left(-0.01,0.01\right)}{0.00}$ & 1.22\tabularnewline
victim of crime & $\underset{\left(-0.01,0.02\right)}{0.00}$ & 1.09 & $\underset{\left(-0.01,0.02\right)}{0.01}$ & 1.16\tabularnewline
promoted at work & $\underset{\left(-0.01,0.02\right)}{0.00}$ & 1.29 & $\underset{\left(-0.01,0.01\right)}{0.00}$ & 1.26\tabularnewline
back with spouse & $\underset{\left(-0.02,0.04\right)}{0.00}$ & 1.26 & $\underset{\left(-0.03,0.03\right)}{-0.00}$ & 1.38\tabularnewline
separated from spouse & $\underset{\left(-0.00,0.03\right)}{0.01}$ & 1.14 & $\underset{\left(-0.00,0.03\right)}{0.01}$ & 1.22\tabularnewline
improvement in financial & $\underset{\left(-0.01,0.01\right)}{-0.00}$ & 1.04 & $\underset{\left(-0.02,0.02\right)}{-0.00}$ & 1.10\tabularnewline
worsening in financial & $\underset{\left(-0.01,0.02\right)}{0.01}$ & 1.26 & $\underset{\left(-0.01,0.03\right)}{0.01}$ & 1.29 \tabularnewline
\hline
Unconditional Prob. & 0.065 &  &  & \tabularnewline
\hline
\end{tabular}
\end{table}

\begin{table}[H]
\caption{Average Partial Effects for the probability of binge drinking and
mental health for major life events variables for Female. Symbol $^{*}$ denotes statistical significance. \label{tab:Average-Partial-Effects for female}}

\centering{}%
\begin{tabular}{ccccc}
\hline
Variables & Probit & IACT & Gaussian & IACT\tabularnewline
\hline
\hline
gave birth & $\underset{\left(-0.03,-0.01\right)}{-0.02}$ & 1.60 & $\underset{\left(-0.04,-0.01\right)}{-0.03}$ & 1.56\tabularnewline
death of a friend & $\underset{\left(-0.00,0.02\right)}{0.01}$ & 1.53 & $\underset{\left(0.00,0.02\right)}{0.01}$ & 1.66\tabularnewline
death of a spouse/child & $\underset{\left(-0.02,0.05\right)}{0.01}$ & 1.32 & $\underset{\left(-0.03,0.04\right)}{0.01}$ & 1.51\tabularnewline
personal injury & $\underset{\left(-0.00,0.02\right)}{0.01}$ & 1.49  & $\underset{\left(-0.00,0.02\right)}{0.01}$ & 1.59\tabularnewline
getting married & $\underset{\left(-0.01,0.02\right)}{0.00}$ & 1.36 & $\underset{\left(-0.01,0.03\right)}{0.01}$ & 1.55\tabularnewline
changed residence & $\underset{\left(0.00,0.01\right)^{*}}{0.01}$ & 1.47 & $\underset{\left(0.00,0.02\right)^{*}}{0.01}$ & 1.75\tabularnewline
victim of crime & $\underset{\left(-0.01,0.01\right)}{-0.00}$ & 1.46 & $\underset{\left(-0.02,0.01\right)}{-0.00}$ & 1.56\tabularnewline
promoted at work & $\underset{\left(0.01,0.03\right)^{*}}{0.02}$ & 1.45 & $\underset{\left(0.01,0.03\right)^{*}}{0.02}$ & 1.59\tabularnewline
back with spouse & $\underset{\left(-0.02,0.03\right)}{-0.00}$ & 1.42 & $\underset{\left(-0.03,0.03\right)}{-0.00}$ & 1.64\tabularnewline
separated from spouse & $\underset{\left(0.01,0.04\right)^{*}}{0.03}$ & 1.36 & $\underset{\left(0.01,0.04\right)^{*}}{0.03}$ & 1.60\tabularnewline
improvement in financial & $\underset{\left(-0.01,0.02\right)}{0.01}$ & 1.49 & $\underset{\left(-0.01,0.02\right)}{0.01}$ & 1.44\tabularnewline
worsening in financial & $\underset{\left(0.01,0.05\right)^{*}}{0.03}$ & 1.61 & $\underset{\left(0.01,0.06\right)^{*}}{0.03}$ & 1.74\tabularnewline
\hline
Unconditional Prob. & 0.058 &  &  & \tabularnewline
\hline
\end{tabular}
\end{table}

\section{Conclusions}\label{S: conclusions}

Based on recent
advances in Particle Markov chain Monte Carlo (PMCMC), we demonstrate an
approach to estimating flexible model specifications
for multivariate outcomes using panel data. We propose a particle Metropolis within Gibbs (PMwG)  sampling scheme
for Bayesian inference of a flexible
for multivariate outcomes using panel data and show that this sampler is more efficient than competing methods.

The panel data methods we develop in this paper also accommodate
a mix of discrete and continuous outcomes and in doing so avoid the
common approach of reducing all outcomes to binary variables so that
a multivariate probit approach is possible. We demonstrate in our application that joint modelling of alcohol consumption and mental
health often gave only slightly different results after discretising
the outcomes. But given that more general specifications better reflect
the discrete outcomes of alcohol consumption and the continuous mental
health measure, there is an argument that the bivariate probit model
is potentially masking important features of the relationship.

The results in the application are somewhat surprising. Specifying
and comparing different copulas was motivated by the belief that the
dependence structure between excessive alcohol consumption and poor
mental health might potentially be very different in the tails of
the distributions. The results indicate that this is not the case
in that all three copulas provide qualitatively similar results. Moreover,
they indicate that the relationship between alcohol consumption and
mental health is weak which is a key reason why differences did not
emerge across different copulas. While we have not estimated a formal
model allowing two-way causality between excessive alcohol consumption
and mental health, if such effects exist then we would expect them
to manifest themselves in a positive relationship in our joint estimation.
Not finding such a relationship is possibly evidence that the causal
effects running both ways between excessive alcohol consumption and
mental health are indeed weak or even non-existent. This is not inconsistent
with the existing literature where evidence is mixed; see for example
\citet{Boden:2011}. Another possibility is that there are causal
effects but they relate to particular subgroups of the population
and our models are insufficiently rich to capture the heterogeneity
in these effects. We have conducted all analyses for males and females
separately and found some differences across these groups but it may
be that other sources of heterogeneity may be associated with unobservable
rather than observable individual features. We leave this interesting line of work for future research.

\bibliographystyle{Chicago}
\bibliography{references_v1}

\section*{Appendices}
\begin{appendix}

\section{Proofs of Results}  \label{sec:lemma1}

The first lemma shows that the estimate $\widehat{p}_{N}(\bs y |\bs \theta)$ given in \cref{eq:estimated likelihood} is an
unbiased estimate of the likelihood $p(\bs y |\bs \theta)$.
\begin{lemma}
\label{lemma1}
$E \left\{ \widehat{p}_{N}(\bs y |\bs \theta) \right\}   = p(\bs y |\bs \theta)$.
\begin{proof}
From \cref{ass: use of IS} and Steps 1 and 2 of \cref{alg: IS sampling alg},
$E \left( w_{i}^{j} \right)  = p(\bs y_i |\bs \theta)$, and
hence the result follows from equations \eqref{eq:independence likelihood}), \eqref{eq:particle dist}) and
\eqref{eq:estimated likelihood}).
\end{proof}
\end{lemma}

\begin{lemma}
\label{lemma2}
The marginal distribution of $\widetilde{\pi}_{N}\left(\boldsymbol{k},\boldsymbol{\alpha}_{1:P}^{\boldsymbol{k}},\boldsymbol{\theta}\right)$
is given by
\[
\widetilde{\pi}_{N}\left(\boldsymbol{k},\boldsymbol{\alpha}_{1:P}^{\boldsymbol{k}},\boldsymbol{\theta}\right)=\frac{\pi\left(\boldsymbol{\theta},\boldsymbol{\alpha}_{1:P}^{\boldsymbol{k}}\right)}{N^{P}}.
\]
\begin{proof}
We integrate the target density $\widetilde{\pi}_{N}\left(\boldsymbol{k},\boldsymbol{\alpha}_{1:P}^{1:N},\boldsymbol{\theta}\right)$
over $\boldsymbol{\alpha}_{1:P}^{\left(-\boldsymbol{k}\right)}$
\begin{eqnarray*}
\widetilde{\pi}_{N}\left(\boldsymbol{k},\boldsymbol{\alpha}_{1:P}^{\boldsymbol{k}},\boldsymbol{\theta}\right)=\int\widetilde{\pi}_{N}\left(\boldsymbol{k},\boldsymbol{\alpha}_{1:P}^{1:N},\boldsymbol{\theta}\right)d\boldsymbol{\alpha}_{1:P}^{\left(-\boldsymbol{k}\right)}=\frac{\pi\left(\boldsymbol{\theta},\boldsymbol{\alpha}_{1:P}^{\boldsymbol{k}}\right)}{N^{P}}.
\end{eqnarray*}
\end{proof}

\end{lemma}

\begin{proof}[Proof of Theorem~\ref{thm:convergence of pmmh}]
The proof follows from Assumption \ref{ass: use of IS}, Lemmas \ref{lemma1} and \ref{lemma2} and Theorem 1 in \citet{Andrieu:2009}.
\end{proof}

\begin{proof}[Proof of Theorem~\ref{thm: converg of PMwG}]
The proof follows the approach in Theorem 5 in \citet[pg. 300]{Andrieu:2010}.
The algorithm is a Metropolis within Gibbs sampler targeting \Cref{eq: expanded target density}.
Hence we focus on establishing irreducibility and aperiodicity.
It will be convenient to use the notation
$
\widetilde{\pi}_{N}\left(\boldsymbol{k},\boldsymbol{\alpha}_{1:P}^{1:N},\boldsymbol{\theta}\right)=\widetilde{\pi}_{N}\left(\boldsymbol{k},\boldsymbol{\alpha}_{1:P}^{\boldsymbol{k}},\boldsymbol{\alpha}_{1:P}^{-\boldsymbol{k}},\boldsymbol{\theta}\right)
$
to partition the particles $\boldsymbol{\alpha}_{1:P}^{1:N}$ into $\boldsymbol{\alpha}_{1:P}^{\boldsymbol{k}}$ and $\boldsymbol{\alpha}_{1:P}^{-\boldsymbol{k}}$
which are the particles selected and not selected by the indices $\boldsymbol{k}$ respectively.

Let $\boldsymbol{k} \in \{1, \ldots, N\}^P$, $D \in {\cal B}\left( {\cal R}^P \right)$, $E \in {\cal B} \left( {\cal R}^{(N-1) P} \right)$
and $F \in {\cal B} \left( \Theta \right)$ be such that
$\widetilde{\pi}_{N}\left( \{\boldsymbol{k}\} \times D \times E \times F \right) > 0$.

From Assumption \ref{ass: use of IS} it is possible to show that accessible sets for the Metropolis within Gibbs sampler
are also marginally accessible by the particle Metropolis within Gibbs sampler.
From this and Assumption \ref{ass: gibbs1}, we deduce that  there is a finite $j > 0$ such that
${\cal L}_{PMwG} \left\{ \left( \boldsymbol{K}(j), \boldsymbol{\alpha}_{1:P}(j), \boldsymbol{\theta}(j) \right)
\in \{\boldsymbol{k}  \}\times D \times F  \right\} > 0$.

Now because Step 2 consists of a Gibbs step using $\widetilde{\pi}_{N}(\cdot)$, we deduce that
\begin{eqnarray*}
{\cal L}_{PMwG} \left\{ \left( \boldsymbol{K}(j), \boldsymbol{\alpha}_{1:P}^{\boldsymbol{K}(j)},\boldsymbol{\alpha}_{1:P}^{-\boldsymbol{K}(j)} \boldsymbol{\theta}(j) \right)
\in \{\boldsymbol{k}  \}\times D \times E \times F  \right\} > 0
\end{eqnarray*}
and the irreducibility of the PMwG samper follows.

Aperiodicity can be proved by contradiction since, if the PMwG sample is periodic then from Assumption \ref{ass: use of IS} so is the MwG sampler, which contradicts Assumption \ref{ass: gibbs1}.
The result now follows from
Theorem 1 of \cite{tierney1994}. 
\end{proof}

\section{Empirical Results \label{sec:Empirical-Results}}

\begin{table}[H]
\caption{Estimation results for male Individual Effects and Dependence Parameters.
Posterior mean estimates with 95\% credible intervals (in brackets).
\label{tab:Estimation-Results-for Male dependence parameters}}

\centering{}%
\begin{tabular}{ccccc}
\hline
 & Gaussian & Clayton & Gumbel & Probit\tabularnewline
\hline
$\rho_{\alpha}$ & $\underset{\left(-0.03,0.10\right)}{0.03}$ & $\underset{\left(-0.04,0.09\right)}{0.02}$ & $\underset{\left(-0.03,0.10\right)}{0.03}$ & $\underset{\left(0.03,0.15\right)}{0.09}$\tabularnewline
Kendall tau & $\underset{\left(-0.02,0.07\right)}{0.02}$ & $\underset{\left(-0.03,0.06\right)}{0.01}$ & $\underset{\left(-0.02,0.06\right)}{0.02}$ & $\underset{\left(0.02,0.09\right)}{0.06}$\tabularnewline
\hline
$\theta_{dep}$ & $\underset{\left(-0.07,0.11\right)}{0.02}$ & $\underset{\left(0.11,0.30\right)}{0.19}$ & $\underset{\left(1.00,1.06\right)}{1.02}$ & $\underset{\left(-0.10,0.09\right)}{-0.00}$\tabularnewline
Kendall tau & $\underset{\left(-0.04,0.07\right)}{0.01}$ & $\underset{\left(0.05,0.13\right)}{0.09}$ & $\underset{\left(0.00,0.06\right)}{0.02}$ & $\underset{\left(-0.06,0.06\right)}{-0.00}$\tabularnewline
Lower/Upper Tail & NA & $\underset{\left(0.00,0.10\right)}{0.03}$ & $\underset{\left(0.00,0.08\right)}{0.03}$ & NA\tabularnewline
\hline
$\tau_{1}^{2}$ & $\underset{\left(3.63,5.04\right)}{4.31}$ & $\underset{\left(3.60,5.09\right)}{4.29}$ & $\underset{\left(3.71,5.15\right)}{4.37}$ & $\underset{\left(3.66,5.07\right)}{4.34}$\tabularnewline
$\tau_{2}^{2}$ & $\underset{\left(0.40,0.47\right)}{0.43}$ &  $\underset{\left(0.40,0.47\right)}{0.43}$ & $\underset{\left(0.40,0.47\right)}{0.44}$ & $\underset{\left(2.45,3.09\right)}{2.76}$\tabularnewline
\hline
\end{tabular}
\end{table}

\begin{table}[H]
\caption{Estimation results for female Individual Effects and Dependence Parameters.
Posterior mean estimates with 95\% credible intervals (in brackets).
\label{tab:Estimation-Results-for feMale dependence parameters}}

\centering{}%
\begin{tabular}{ccccc}
\hline
 & Gaussian & Clayton & Gumbel & Probit\tabularnewline
\hline
$\rho_{\alpha}$ & $\underset{\left(-0.06,0.08\right)}{0.01}$ & $\underset{\left(-0.08,0.06\right)}{-0.01}$ & $\underset{\left(-0.06,0.07\right)}{0.00}$ & $\underset{\left(-0.04,0.08\right)}{0.02}$\tabularnewline
Kendall tau & $\underset{\left(-0.04,0.05\right)}{0.00}$ & $\underset{\left(-0.05,0.04\right)}{-0.01}$ & $\underset{\left(-0.04,0.04\right)}{0.00}$ & $\underset{\left(-0.02,0.05\right)}{0.01}$\tabularnewline
\hline
$\theta_{dep}$ & $\underset{\left(-0.07,0.11\right)}{0.02}$ & $\underset{\left(0.11,0.29\right)}{0.19}$ & $\underset{\left(1.00,1.06\right)}{1.02}$ & $\underset{\left(-0.08,0.08\right)}{-0.00}$\tabularnewline
Kendall tau & $\underset{\left(-0.05,0.07\right)}{0.01}$ & $\underset{\left(0.05,0.13\right)}{0.09}$ & $\underset{\left(0.00,0.05\right)}{0.02}$ & $\underset{\left(-0.05,0.05\right)}{-0.00}$\tabularnewline
Lower/Upper tail & NA & $\underset{\left(0.00,0.09\right)}{0.03}$ & $\underset{\left(0.00,0.07\right)}{0.03}$ & NA\tabularnewline
\hline
$\tau_{1}^{2}$ & $\underset{\left(3.35,4.62\right)}{3.93}$ & $\underset{\left(3.31,4.58\right)}{3.90}$ & $\underset{\left(3.37,4.69\right)}{3.99}$ & $\underset{\left(3.38,4.63\right)}{3.96}$\tabularnewline
$\tau_{2}^{2}$ & $\underset{\left(0.34,0.40\right)}{0.37}$ & $\underset{\left(0.34,0.40\right)}{0.37}$ & $\underset{\left(0.34,0.40\right)}{0.37}$ & $\underset{\left(1.91,2.36\right)}{2.13}$\tabularnewline
\hline
\end{tabular}
\end{table}

\begin{table}[H]
\caption{Estimation results for male Binge/Excessive Drinking $\left(y_{1}\right)$
(balanced panel). Posterior mean estimates with 95\% credible
intervals (in brackets). \label{tab:Estimation-Results-for male binge }}

\centering{}%
\begin{tabular}{ccccc}
\hline
 & Gaussian & Probit & Clayton & Gumbel\tabularnewline
\hline
university/degree & $\underset{\left(-0.97,-0.41\right)}{-0.69}$ & $\underset{\left(-0.97,-0.40\right)}{-0.69}$ & $\underset{\left(-0.96,-0.40\right)}{-0.68}$ & $\underset{\left(-0.97,-0.40\right)}{-0.68}$\tabularnewline
diploma/certificate & $\underset{\left(-0.22,0.22\right)}{-0.01}$ & $\underset{\left(-0.23,0.22\right)}{0.00}$ & $\underset{\left(-0.21,0.23\right)}{0.01}$ & $\underset{\left(-0.22,0.23\right)}{0.00}$\tabularnewline
married & $\underset{\left(-0.61,-0.21\right)}{-0.41}$ & $\underset{\left(-0.61,-0.21\right)}{-0.41}$ & $\underset{\left(-0.62,-0.22\right)}{-0.41}$ & $\underset{\left(-0.62,-0.22\right)}{-0.42}$\tabularnewline
income & $\underset{\left(-0.11,0.02\right)}{-0.05}$ & $\underset{\left(-0.11,0.02\right)}{-0.05}$ & $\underset{\left(-0.11,0.01\right)}{-0.05}$ & $\underset{\left(-0.11,0.02\right)}{-0.05}$\tabularnewline
num. child & $\underset{\left(-0.27,-0.05\right)}{-0.16}$ & $\underset{\left(-0.27,-0.05\right)}{-0.16}$ & $\underset{\left(-0.27,-0.05\right)}{-0.16}$ & $\underset{\left(-0.27,-0.05\right)}{-0.16}$\tabularnewline
gave birth & $\underset{\left(-0.66,-0.05\right)}{-0.35}$ & $\underset{\left(-0.66,-0.06\right)}{-0.35}$ & $\underset{\left(-0.66,-0.04\right)}{-0.34}$ & $\underset{\left(-0.67,-0.05\right)}{-0.36}$\tabularnewline
death of a friend & $\underset{\left(-0.12,0.26\right)}{0.07}$ & $\underset{\left(-0.12,0.25\right)}{0.07}$ & $\underset{\left(-0.12,0.25\right)}{0.07}$ & $\underset{\left(-0.12,0.26\right)}{0.07}$\tabularnewline
death of a spouse/child & $\underset{\left(-0.62,0.98\right)}{0.20}$ & $\underset{\left(-0.63,1.00\right)}{0.21}$ & $\underset{\left(-0.74,1.00\right)}{0.17}$ & $\underset{\left(-0.74,1.01\right)}{0.17}$\tabularnewline
personal injury & $\underset{\left(-0.08,0.32\right)}{0.12}$ & $\underset{\left(-0.08,0.32\right)}{0.12}$ & $\underset{\left(-0.06,0.34\right)}{0.14}$ & $\underset{\left(-0.08,0.33\right)}{0.13}$\tabularnewline
getting married & $\underset{\left(-0.61,0.15\right)}{-0.22}$ & $\underset{\left(-0.60,0.14\right)}{-0.22}$ & $\underset{\left(-0.61,0.13\right)}{-0.23}$ &  $\underset{\left(-0.61,0.14\right)}{-0.23}$\tabularnewline
changed residence & $\underset{\left(-0.11,0.20\right)}{0.04}$ & $\underset{\left(-0.12,0.19\right)}{0.04}$ & $\underset{\left(-0.10,0.20\right)}{0.05}$ & $\underset{\left(-0.12,0.20\right)}{0.04}$\tabularnewline
victim of crime & $\underset{\left(-0.31,0.24\right)}{-0.04}$ & $\underset{\left(-0.32,0.22\right)}{-0.04}$ & $\underset{\left(-0.30,0.24\right)}{-0.03}$ & $\underset{\left(-0.33,0.23\right)}{-0.05}$\tabularnewline
promoted at work & $\underset{\left(-0.14,0.30\right)}{0.08}$ & $\underset{\left(-0.14,0.30\right)}{0.08}$ & $\underset{\left(-0.15,0.30\right)}{0.08}$ & $\underset{\left(-0.14,0.30\right)}{0.08}$\tabularnewline
back with spouse & $\underset{\left(-0.84,0.40\right)}{-0.21}$ & $\underset{\left(-0.84,0.41\right)}{-0.21}$ & $\underset{\left(-0.84,0.42\right)}{-0.20}$ & $\underset{\left(-0.85,0.39\right)}{-0.21}$\tabularnewline
separated from spouse & $\underset{\left(-0.26,0.33\right)}{0.03}$ & $\underset{\left(-0.27,0.31\right)}{0.02}$ & $\underset{\left(-0.27,0.32\right)}{0.03}$ & $\underset{\left(-0.27,0.31\right)}{0.02}$\tabularnewline
improvement in financial & $\underset{\left(-0.31,0.30\right)}{-0.00}$ & $\underset{\left(-0.31,0.32\right)}{-0.00}$ & $\underset{\left(-0.32,0.31\right)}{-0.00}$ & $\underset{\left(-0.32,0.32\right)}{-0.00}$\tabularnewline
worsening in financial & $\underset{\left(-0.43,0.21\right)}{-0.11}$ & $\underset{\left(-0.45,0.20\right)}{-0.12}$ & $\underset{\left(-0.43,0.21\right)}{-0.10}$ & $\underset{\left(-0.44,0.20\right)}{-0.12}$\tabularnewline
\hline
$\min_{1:23}IACT\left(\beta_{1i}\right)$ & 1.00 & 1.00 & 1.00 & 1.13\tabularnewline
$\max_{1:23}IACT\left(\beta_{1i}\right)$ & 7.58 & 8.19 & 8.67 & 8.64\tabularnewline
${\rm mean}_{1:23}\left(\beta_{1i}\right)$ & 2.55 & 2.55 & 2.59 & 2.60\tabularnewline
\hline
\end{tabular}
\end{table}

\begin{table}[H]
\caption{Estimation results for male Mental Health score $\left(y_{2}\right)$
(balanced panel). Posterior mean estimates with 95\% credible
intervals (in brackets). \label{tab:Estimation-Results-for male mental health}}

\centering{}%
\begin{tabular}{ccccc}
\hline
 & Gaussian & Probit & Clayton & Gumbel\tabularnewline
\hline
university/degree & $\underset{\left(-0.14,0.03\right)}{-0.05}$ & $\underset{\left(-0.29,0.07\right)}{-0.11}$ & $\underset{\left(-0.13,0.03\right)}{-0.05}$ & $\underset{\left(-0.14,0.03\right)}{-0.05}$\tabularnewline
diploma/certificate & $\underset{\left(-0.10,0.05\right)}{-0.02}$ & $\underset{\left(-0.22,0.10\right)}{-0.06}$ & $\underset{\left(-0.09,0.05\right)}{-0.02}$ & $\underset{\left(-0.10,0.05\right)}{-0.03}$\tabularnewline
married & $\underset{\left(-0.12,0.01\right)}{-0.05}$ & $\underset{\left(-0.28,0.01\right)}{-0.13}$ & $\underset{\left(-0.12,0.01\right)}{-0.05}$ & $\underset{\left(-0.12,0.01\right)}{-0.05}$\tabularnewline
income & $\underset{\left(-0.02,0.03\right)}{0.01}$ & $\underset{\left(-0.07,0.02\right)}{-0.02}$ & $\underset{\left(-0.02,0.03\right)}{0.01}$ & $\underset{\left(-0.02,0.04\right)}{0.01}$\tabularnewline
num. child & $\underset{\left(-0.03,0.06\right)}{0.02}$ & $\underset{\left(-0.05,0.10\right)}{0.03}$ & $\underset{\left(-0.03,0.06\right)}{0.02}$ & $\underset{\left(-0.03,0.06\right)}{0.02}$\tabularnewline
gave birth & $\underset{\left(-0.05,0.19\right)}{0.07}$ & $\underset{\left(-0.11,0.30\right)}{0.09}$ & $\underset{\left(-0.05,0.19\right)}{0.07}$ & $\underset{\left(-0.05,0.20\right)}{0.07}$\tabularnewline
death of a friend & $\underset{\left(-0.05,0.10\right)}{0.02}$ & $\underset{\left(-0.07,0.19\right)}{0.06}$ & $\underset{\left(-0.05,0.10\right)}{0.02}$ & $\underset{\left(-0.05,0.09\right)}{0.02}$\tabularnewline
death of a spouse/child & $\underset{\left(-0.14,0.49\right)}{0.18}$ & $\underset{\left(0.02,1.11\right)}{0.56}$ & $\underset{\left(-0.13,0.50\right)}{0.18}$ & $\underset{\left(-0.14,0.49\right)}{0.18}$\tabularnewline
personal injury & $\underset{\left(0.30,0.46\right)}{0.38}$ & $\underset{\left(0.48,0.75\right)}{0.62}$ & $\underset{\left(0.30,0.46\right)}{0.38}$ & $\underset{\left(0.30,0.46\right)}{0.38}$\tabularnewline
getting married & $\underset{\left(-0.18,0.15\right)}{-0.02}$ & $\underset{\left(-0.32,0.24\right)}{-0.04}$ & $\underset{\left(-0.18,0.15\right)}{-0.01}$ &  $\underset{\left(-0.18,0.15\right)}{-0.01}$\tabularnewline
changed residence & $\underset{\left(-0.06,0.07\right)}{0.00}$ & $\underset{\left(-0.18,0.05\right)}{-0.07}$ & $\underset{\left(-0.06,0.07\right)}{0.00}$ & $\underset{\left(-0.07,0.07\right)}{0.00}$\tabularnewline
victim of crime & $\underset{\left(0.04,0.27\right)}{0.16}$ & $\underset{\left(-0.01,0.40\right)}{0.19}$ & $\underset{\left(0.04,0.27\right)}{0.16}$ & $\underset{\left(0.04,0.27\right)}{0.15}$\tabularnewline
promoted at work & $\underset{\left(-0.12,0.06\right)}{-0.03}$ & $\underset{\left(-0.12,0.19\right)}{0.04}$ & $\underset{\left(-0.12,0.06\right)}{-0.03}$ & $\underset{\left(-0.12,0.06\right)}{-0.03}$\tabularnewline
back with spouse & $\underset{\left(-0.12,0.40\right)}{0.14}$ & $\underset{\left(-0.11,0.82\right)}{0.36}$ & $\underset{\left(-0.12,0.41\right)}{0.14}$ & $\underset{\left(-0.13,0.40\right)}{0.13}$\tabularnewline
separated from spouse & $\underset{\left(0.10,0.37\right)}{0.24}$ & $\underset{\left(0.13,0.60\right)}{0.36}$ & $\underset{\left(0.10,0.37\right)}{0.23}$ & $\underset{\left(0.10,0.37\right)}{0.24}$\tabularnewline
improvement in financial & $\underset{\left(-0.13,0.11\right)}{-0.01}$ & $\underset{\left(-0.27,0.16\right)}{-0.05}$ & $\underset{\left(-0.14,0.11\right)}{-0.01}$ & $\underset{\left(-0.13,0.12\right)}{-0.01}$\tabularnewline
worsening in financial & $\underset{\left(0.24,0.49\right)}{0.37}$ & $\underset{\left(0.24,0.69\right)}{0.47}$ & $\underset{\left(0.24,0.49\right)}{0.37}$ & $\underset{\left(0.24,0.49\right)}{0.37}$\tabularnewline
\hline
$\min_{1:23}IACT\left(\beta_{2i}\right)$ & 1.00 & 1.00 & 1.00 & 1.01\tabularnewline
$\max_{1:23}IACT\left(\beta_{2i}\right)$ & 3.31 & 9.92 & 3.60 & 2.73\tabularnewline
${\rm mean}_{1:23}\left(\beta_{2i}\right)$ & 1.41 & 2.43 & 1.30 & 1.51\tabularnewline
\hline
\end{tabular}
\end{table}

\begin{table}[H]
\caption{Estimation results for female Binge/Excessive Drinking $\left(y_{1}\right)$.
Posterior mean estimates with 95\% credible intervals (in brackets).
\label{tab:Estimation-Results-for female binge drinking}}

\centering{}%
\begin{tabular}{ccccc}
\hline
 & Gaussian & Probit & Clayton & Gumbel\tabularnewline
\hline
university/degree & $\underset{\left(-0.61,-0.16\right)}{-0.38}$ & $\underset{\left(-0.61,-0.16\right)}{-0.39}$ & $\underset{\left(-0.62,-0.18\right)}{-0.40}$ & $\underset{\left(-0.63,-0.17\right)}{-0.40}$\tabularnewline
diploma/certificate & $\underset{\left(-0.30,0.13\right)}{-0.09}$ & $\underset{\left(-0.30,0.13\right)}{-0.09}$ & $\underset{\left(-0.31,0.12\right)}{-0.09}$ & $\underset{\left(-0.31,0.12\right)}{-0.09}$\tabularnewline
married & $\underset{\left(-0.68,-0.35\right)}{-0.51}$ & $\underset{\left(-0.69,-0.35\right)}{-0.52}$ & $\underset{\left(-0.70,-0.36\right)}{-0.53}$ & $\underset{\left(-0.70,-0.36\right)}{-0.53}$\tabularnewline
income & $\underset{\left(0.05,0.19\right)}{0.12}$ & $\underset{\left(0.05,0.20\right)}{0.12}$ & $\underset{\left(0.05,0.20\right)}{0.13}$ & $\underset{\left(0.05,0.20\right)}{0.13}$\tabularnewline
num. child & $\underset{\left(-0.22,0.00\right)}{-0.11}$ & $\underset{\left(-0.22,0.00\right)}{-0.11}$ & $\underset{\left(-0.22,0.01\right)}{-0.10}$ & $\underset{\left(-0.22,0.00\right)}{-0.11}$\tabularnewline
gave birth & $\underset{\left(-1.16,-0.49\right)}{-0.81}$ & $\underset{\left(-1.16,-0.48\right)}{-0.81}$ & $\underset{\left(-1.15,-0.47\right)}{-0.81}$ & $\underset{\left(-1.18,-0.48\right)}{-0.83}$\tabularnewline
death of a friend & $\underset{\left(-0.00,0.37\right)}{0.19}$ & $\underset{\left(-0.00,0.36\right)}{0.18}$ & $\underset{\left(-0.01,0.36\right)}{0.17}$ & $\underset{\left(-0.01,0.37\right)}{0.18}$\tabularnewline
death of a spouse/child & $\underset{\left(-0.90,0.56\right)}{-0.16}$ & $\underset{\left(-0.92,0.56\right)}{-0.16}$ & $\underset{\left(-0.95,0.56\right)}{-0.17}$ & $\underset{\left(-0.93,0.59\right)}{-0.16}$\tabularnewline
personal injury & $\underset{\left(-0.43,0.03\right)}{-0.20}$ & $\underset{\left(-0.43,0.03\right)}{-0.20}$ & $\underset{\left(-0.42,0.04\right)}{-0.18}$ & $\underset{\left(-0.44,0.03\right)}{-0.20}$\tabularnewline
getting married & $\underset{\left(-0.17,0.55\right)}{0.19}$ & $\underset{\left(-0.17,0.55\right)}{0.20}$ & $\underset{\left(-0.16,0.55\right)}{0.20}$ &  $\underset{\left(-0.17,0.56\right)}{0.20}$\tabularnewline
changed residence & $\underset{\left(0.04,0.32\right)}{0.18}$ & $\underset{\left(0.04,0.31\right)}{0.17}$ & $\underset{\left(0.04,0.32\right)}{0.18}$ & $\underset{\left(0.04,0.32\right)}{0.18}$\tabularnewline
victim of crime & $\underset{\left(-0.43,0.15\right)}{-0.14}$ & $\underset{\left(-0.44,0.15\right)}{-0.14}$ & $\underset{\left(-0.43,0.15\right)}{-0.13}$ & $\underset{\left(-0.44,0.14\right)}{-0.15}$\tabularnewline
promoted at work & $\underset{\left(0.12,0.53\right)}{0.33}$ & $\underset{\left(0.12,0.53\right)}{0.33}$ & $\underset{\left(0.11,0.53\right)}{0.32}$ & $\underset{\left(0.12,0.53\right)}{0.33}$\tabularnewline
back with spouse & $\underset{\left(-0.91,0.29\right)}{-0.30}$ & $\underset{\left(-0.88,0.29\right)}{-0.31}$ & $\underset{\left(-0.91,0.30\right)}{-0.30}$ & $\underset{\left(-0.93,0.27\right)}{-0.32}$\tabularnewline
separated from spouse & $\underset{\left(0.11,0.65\right)}{0.38}$ & $\underset{\left(0.11,0.65\right)}{0.38}$ & $\underset{\left(0.11,0.65\right)}{0.38}$ & $\underset{\left(0.10,0.65\right)}{0.38}$\tabularnewline
improvement in financial & $\underset{\left(-0.15,0.44\right)}{0.14}$ & $\underset{\left(-0.17,0.44\right)}{0.14}$ & $\underset{\left(-0.16,0.45\right)}{0.14}$ & $\underset{\left(-0.15,0.45\right)}{0.15}$\tabularnewline
worsening in financial & $\underset{\left(-0.10,0.56\right)}{0.24}$ & $\underset{\left(-0.10,0.56\right)}{0.24}$ & $\underset{\left(-0.06,0.61\right)}{0.27}$ & $\underset{\left(-0.08,0.57\right)}{0.25}$\tabularnewline
\hline
$\min_{1:23}IACT\left(\beta_{1i}\right)$ & 1.25 & 1.01 & 1.00 & 1.00\tabularnewline
$\max_{1:23}IACT\left(\beta_{1i}\right)$ & 6.55 & 7.56 & 6.78 & 7.58\tabularnewline
${\rm mean}_{1:23}\left(\beta_{1i}\right)$ & 2.35 & 2.41 & 2.41 & 2.27\tabularnewline
\hline
\end{tabular}
\end{table}

\begin{table}[H]
\caption{Estimation results for female Mental Health score $\left(y_{2}\right)$
(balanced panel) Posterior mean estimates with 95\% credible intervals
(in brackets). \label{tab:Estimation-Results-for female mental health}}

\centering{}%
\begin{tabular}{ccccc}
\hline
 & Gaussian & Probit & Clayton & Gumbel\tabularnewline
\hline
university/degree & $\underset{\left(-0.18,-0.04\right)}{-0.11}$ & $\underset{\left(-0.39,-0.10\right)}{-0.24}$ & $\underset{\left(-0.18,-0.04\right)}{-0.11}$ & $\underset{\left(-0.18,-0.04\right)}{-0.11}$\tabularnewline
diploma/certificate & $\underset{\left(-0.10,0.03\right)}{-0.04}$ & $\underset{\left(-0.21,0.07\right)}{-0.07}$ & $\underset{\left(-0.10,0.03\right)}{-0.04}$ & $\underset{\left(-0.10,0.03\right)}{-0.04}$\tabularnewline
married & $\underset{\left(-0.21,-0.09\right)}{-0.15}$ & $\underset{\left(-0.36,-0.14\right)}{-0.25}$ & $\underset{\left(-0.21,-0.09\right)}{-0.15}$ & $\underset{\left(-0.20,-0.09\right)}{-0.15}$\tabularnewline
income & $\underset{\left(-0.02,0.03\right)}{0.01}$ & $\underset{\left(-0.04,0.04\right)}{0.00}$ & $\underset{\left(-0.02,0.03\right)}{0.01}$ & $\underset{\left(-0.02,0.03\right)}{0.01}$\tabularnewline
num. child & $\underset{\left(-0.04,0.05\right)}{0.01}$ & $\underset{\left(-0.05,0.09\right)}{0.02}$ & $\underset{\left(-0.03,0.06\right)}{0.01}$ & $\underset{\left(-0.03,0.06\right)}{0.01}$\tabularnewline
gave birth & $\underset{\left(0.02,0.26\right)}{0.14}$ & $\underset{\left(0.09,0.47\right)}{0.28}$ & $\underset{\left(0.03,0.26\right)}{0.14}$ & $\underset{\left(0.02,0.25\right)}{0.14}$\tabularnewline
death of a friend & $\underset{\left(-0.04,0.10\right)}{0.03}$ & $\underset{\left(-0.09,0.14\right)}{0.02}$ & $\underset{\left(-0.04,0.10\right)}{0.03}$ & $\underset{\left(-0.04,0.10\right)}{0.03}$\tabularnewline
death of a spouse/child & $\underset{\left(0.01,0.52\right)}{0.27}$ & $\underset{\left(0.08,0.97\right)}{0.52}$ & $\underset{\left(0.01,0.52\right)}{0.26}$ & $\underset{\left(0.01,0.52\right)}{0.27}$\tabularnewline
personal injury & $\underset{\left(0.36,0.51\right)}{0.44}$ & $\underset{\left(0.58,0.85\right)}{0.72}$ & $\underset{\left(0.36,0.51\right)}{0.44}$ & $\underset{\left(0.36,0.52\right)}{0.44}$\tabularnewline
getting married & $\underset{\left(-0.19,0.11\right)}{-0.04}$ & $\underset{\left(-0.41,0.09\right)}{-0.16}$ & $\underset{\left(-0.19,0.12\right)}{-0.04}$ &  $\underset{\left(-0.19,0.11\right)}{-0.04}$\tabularnewline
changed residence & $\underset{\left(-0.04,0.09\right)}{0.03}$ & $\underset{\left(-0.11,0.09\right)}{-0.01}$ & $\underset{\left(-0.03,0.09\right)}{0.03}$ & $\underset{\left(-0.04,0.08\right)}{0.02}$\tabularnewline
victim of crime & $\underset{\left(-0.04,0.18\right)}{0.07}$ & $\underset{\left(-0.08,0.30\right)}{0.11}$ & $\underset{\left(-0.04,0.18\right)}{0.07}$ & $\underset{\left(-0.05,0.18\right)}{0.07}$\tabularnewline
promoted at work & $\underset{\left(-0.08,0.10\right)}{0.01}$ & $\underset{\left(0.02,0.32\right)}{0.17}$ & $\underset{\left(-0.08,0.11\right)}{0.01}$ & $\underset{\left(-0.08,0.10\right)}{0.01}$\tabularnewline
back with spouse & $\underset{\left(-0.01,0.51\right)}{0.25}$ & $\underset{\left(-0.10,0.81\right)}{0.35}$ & $\underset{\left(-0.01,0.52\right)}{0.25}$ & $\underset{\left(-0.01,0.51\right)}{0.26}$\tabularnewline
separated from spouse & $\underset{\left(0.04,0.30\right)}{0.17}$ & $\underset{\left(0.11,0.52\right)}{0.32}$ & $\underset{\left(0.04,0.30\right)}{0.17}$ & $\underset{\left(0.04,0.30\right)}{0.17}$\tabularnewline
improvement in financial & $\underset{\left(-0.14,0.09\right)}{-0.02}$ & $\underset{\left(-0.11,0.26\right)}{0.07}$ & $\underset{\left(-0.14,0.09\right)}{-0.02}$ & $\underset{\left(-0.14,0.08\right)}{-0.02}$\tabularnewline
worsening in financial & $\underset{\left(0.35,0.60\right)}{0.48}$ & $\underset{\left(0.40,0.86\right)}{0.63}$ & $\underset{\left(0.35,0.61\right)}{0.48}$ & $\underset{\left(0.34,0.60\right)}{0.47}$\tabularnewline
\hline
$\min_{1:23}IACT\left(\beta_{2i}\right)$ & 1.23 & 1.00 & 1.00 & 1.00\tabularnewline
$\max_{1:23}IACT\left(\beta_{2i}\right)$ & 2.48 & 5.59 & 3.17 & 2.94\tabularnewline
${\rm mean}_{1:23}\left(\beta_{2i}\right)$ & 1.58 & 2.09 & 1.41 & 1.18\tabularnewline
\hline
\end{tabular}
\end{table}
\end{appendix}
\clearpage
\renewcommand{\theequation}{S\arabic{equation}}
\renewcommand{\thesection}{S\arabic{section}}
\renewcommand{\theproposition}{S\arabic{proposition}}
\renewcommand{\theassumption}{S\arabic{assumption}}
\renewcommand{\thelemma}{S\arabic{lemma}}
\renewcommand{\thecorollary}{S\arabic{corollary}}
\renewcommand{\thealgorithm}{S\arabic{algorithm}}
\renewcommand{\thefigure}{S\arabic{figure}}
\renewcommand{\thetable}{S\arabic{table}}
\renewcommand{\thepage}{S\arabic{page}}
\renewcommand{\thetable}{S\arabic{table}}
\renewcommand{\thepage}{S\arabic{page}}
\setcounter{page}{1}
\setcounter{section}{0}
\setcounter{equation}{0}
\setcounter{algorithm}{0}
\setcounter{table}{0}
\title{\Huge  \sf Online Supplement to \lq Efficient Bayesian estimation for flexible panel models for multivariate outcomes: Impact of life events on mental health and excessive alcohol consumption\rq}

\renewcommand\Authands{ and }

\maketitle

\section{Gaussian, Clayton and Gumbel Copula Models\label{sec:Archimedian-and-Elliptical copulas}}

The Gaussian copula is
\[
C^{\rm Gauss}\left(u_{1},u_{2};\rho \right)=\Phi_{2}\left(\Phi^{-1}\left(u_{1}\right),\Phi^{-1}\left(u_{2}\right)\right),
\]
It can capture both positive and negative dependence
and has the full range $\left(-1,1\right)$ of  pairwise correlations.
where $\Phi_{2}$ is the distribution function of the standard bivariate
normal distribution, $\Phi$ is the distribution function for the
standard univariate normal distribution, and $\theta$ is the dependence
parameter. The dependence structure in the Gaussian copula is symmetric, making
it  unsuitable for data that exhibits strong lower tail
or upper tail dependence.

The baseline model in \Cref{ss: biv probit with random effects}
is a special case of a Gaussian copula where all the univariate marginal distributions
are normally distributed.

The bivariate Clayton copula is
\[
C^{\rm Cl}\left(u_{1},u_{2};\theta\right)=\left(u_{1}^{-\theta}+u_{2}^{-\theta}-1\right)^{-\frac{1}{\theta}}
\]
It can  only
capture positive dependence, although one can reflect a Clayton
copula to model the dependence between $u_{1}$ and $-u_{2}$ instead.
The dependence parameter $\theta$ is defined on the interval of $\left(0,\infty\right)$.
It is suitable for the data which exhibits strong lower tail dependence
and weak upper tail dependence.

The bivariate Gumbel copula is
\[
C^{\rm Gu}\left(u_{1},u_{2};\theta\right)=\exp\left\{ -\left(\left(-\log u_{1}\right)^{\theta}+\left(-\log u_{2}\right)^{\theta}\right)^{1/\theta}\right\}
\]
For the Gumbel copula, the dependence parameter $\theta$ is defined
on the interval $\left[1,\infty\right)$, where $1$ represents the
independence case. The Gumbel copula only captures positive dependence.
It is suitable for data which exhibits strong upper tail dependence
and weak lower tail dependence.

\Cref{fig: copula plots} plots 10, 000 draws from each of the three copula models.

The conditional copula distribution functions for the bivariate copulas
used  in our article can be computed in closed form and are given
by
\[
C_{1|2}^{\rm Gauss}\left(u_{1}|u_{2};\rho \right)=\Phi\left(\frac{\Phi^{-1}\left(u_{1}\right)-\rho\Phi^{-1}\left(u_{2}\right)}{\sqrt{1-\rho^{2}}}\right)
\]

\[
C_{1|2}^{\rm Cl}\left(u_{1}|u_{2};\theta\right)=u_{2}^{-\theta -1}\left(u_{1}^{-\theta}+u_{2}^{-\theta}-1\right)^{-1-1/\theta}
\]

\begin{align*}
C_{1|2}^{\rm Gu}\left(u_{1}|u_{2};\theta\right) & =  C^{\rm Gu}\left(u_{1},u_{2};\theta\right)
\frac{1}{u_{2}}
\left(-\log u_{2}\right)^{\theta-1}\\
 &  \left\{ \left(-\log u_{1}\right)^{\theta}+\left(-\log u_{2}\right)^{\theta}\right\} ^{1/\theta-1}
\end{align*}

\begin{figure}[!h]
\centering
\caption{Left panel: 10000 draws from a Gaussian copula with $\theta=0.8$; center panel:
10000 draws from a Clayton copula with $\theta=6$; right panel: 10000 draws from a Gumbel copula with $\theta=6$
\label{fig: copula plots}}
\begin{tabular}{ccc}
\includegraphics[width=5cm,height=5cm]{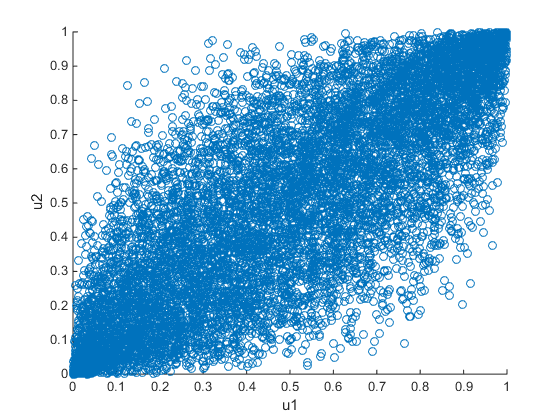} &
\includegraphics[width=5cm,height=5cm]{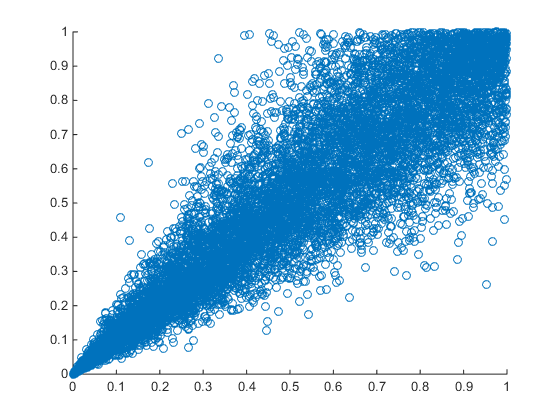} &
\includegraphics[width=5cm,height=5cm]{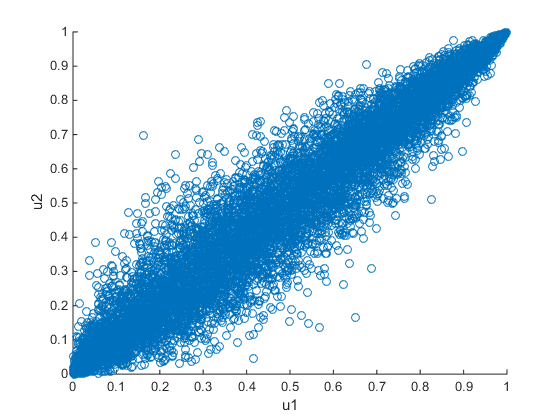}
\end{tabular}
\end{figure}

\section{Dependence Measures\label{sec:Measures-of-Dependence}}

We use Kendall's tau and upper and lower measures of tail dependence
to compare the the dependence structures implied
by different copula models because the Pearson (linear) correlation coefficient is not a good measure of general dependence
between two random variables as
 it only detects linear dependence.

Kendall's tau \citep[][p. 54]{Joe:2015}, is a popular measure of the degree
of concordance between two random variables. Let  $\left(U_{1},V_{1}\right)$
and $\left(U_{2},V_{2}\right)$ be two draws from the joint distribution of $U$
and $V$. Kendall's tau  is defined as
\[
\kappa_{\tau}:=\Pr\left[\left(U_{1}-U_{2}\right)\left(V_{1}-V_{2}\right)>0\right]-\Pr\left[\left(U_{1}-U_{2}\right)\left(V_{1}-V_{2}\right)<0\right]
\]
The value of $\kappa_{\tau}$ can vary between $-1$ to $1$ and is zero
if the two random variables are independent. The Gumbel and Clayton copulas only capture
positive dependence so that $0 \leq \kappa_{\tau}\leq 1 $
For the Gaussian copula, $-1\leq \kappa_{\tau}\leq 1 $. For the copula models we
consider,  $\kappa_{\tau}$ can be computed in closed form as a function
of its copula parameters.
\begin{align*}
\kappa_{\tau}^ {\rm Gauss}& =\frac{2}{\pi}\arcsin\left(\rho \right), \kappa_{\tau}^ {\rm Cl}=\frac{\theta_{Cl}}{\theta+2},
\kappa_{\tau}^ {\rm Gu}=1-\theta_{\rm Gu}^{-1}
\end{align*}

In many cases, the concordance between extreme (tail) values of random
variables is of interest, i.e.  the clustering of extreme events in the upper or lower tails.
For example, suppose we are interested in the relationship
between poor mental health and excessive/binge alcohol consumption
or good mental health and no alcohol consumption. This requires a
dependence measure for the upper and lower tails of the bivariate distribution.
In this case, measures of asymmetric dependence are often based on
conditional probabilities.
The lower and upper tail dependence measures are defined as \citep[][p. 62]{Joe:2015},
\[
\lambda^{U}:=\underset{\alpha\uparrow1}{\lim}\Pr\left(U_{1}>\alpha|U_{2}>\alpha\right)\quad {\rm and} \quad
\lambda^{L}:=\underset{\alpha\downarrow0}{\lim}\Pr\left(U_{1}<\alpha|U_{2}<\alpha\right)
\]
If $\lambda^{U}=0$, then the copula is said to have no upper tail dependence, and if $\lambda^{L}=0$ then the copula is said to have no lower tail dependence. The Gaussian copula has no lower or upper tail dependence.
For the Clayton copula,
$
\lambda^{L}=2^{-1/\theta} > 0$  and $\lambda^{U} = 0 $ so the Clayton copula has no upper tail dependence.
For the Gumbel copula,
$
\lambda^{U}=2-2^{1/\theta}$ and  $\lambda^{L} = 0 $. Hence the Gumbel copula has no lower tail dependence and it has upper tail dependence if and only if
$\theta  \neq 1 $.

\section{Simulation Mixed Discrete Linear Gaussian Regression\label{sec:Simulation-Mixed-Discrete}}
This section provides additional simulation study to compare our particle Metropolis-within-Gibbs approach to the MCMC-MH using mixed discrete linear Gaussian regression model given in Section \ref{ss: mixed biv model with random effects}.
The design is similar to the first experiment in Section \ref{sub:TNV} with $n=1000$ and $T=4$,
$x_{1,it},...,x_{10,it}\sim U\left(0,1\right)$, true parameters set
as follows:

$\beta_{1}=\left(-1.5,0.1,-0.2,0.2,-0.2,0.1,-0.2,0.1,-0.1,-0.2,0.2\right)^{'}$,

$\beta_{2}=\left(-0.5,0.1,0.2,-0.2,0.2,0.12,0.2,-0.2,0.12,-0.12,0.12\right)^{'}$,
$\tau_{1}^{2}=1$, $\tau_{2}^{2}=2.5$, and, $\rho_{\epsilon}=\rho_{\alpha}=0.5$.
We only compare MCMC-MH and PG methods for this simulation since they
can be applied more generally to panel data models with random effects.
\Cref{tab:Comparison-of-Different mixed discrete linear Gaussian1,tab:Comparison-of-Different mixed discrete linear Gaussian2,tab:Comparison-of-Different mixed discrete linear Gaussian2-1}
summarise the estimation results and show that the PG is still much
better than MCMC-MH methods.

\begin{table}[H]
\caption{Comparison of Inefficiency Factors (IACT) for the Parameters with
Different Sampling Schemes (PG, data augmentation, MCMC-MH) of Mixed
Discrete-Linear Gaussian regression Simulation with random effects
$P=1000$ and $T=4$. } \label{tab:Comparison-of-Different mixed discrete linear Gaussian1}

\centering{}%
\begin{tabular}{cccccc}
\hline
Param.  & PG & MH1 & MH10 & MH20 & MH50\tabularnewline
\hline
$\beta_{11}$ & $1.58$ & $3.64$ & $1.10$ & $1.34$ & $1.00$\tabularnewline
$\beta_{21}$ & $1.71$ & $4.15$ & $1.11$ & $1.55$ & $1.00$\tabularnewline
$\beta_{31}$ & $1.65$ & $4.49$ & $1.00$ & $1.40$ & $1.00$\tabularnewline
$\beta_{41}$ & $1.82$ & $3.34$ & $1.10$ & $1.53$ & $1.00$\tabularnewline
$\beta_{51}$ & $1.55$ & $3.13$ & $1.17$ & $1.48$ & $1.00$\tabularnewline
$\beta_{61}$ & $1.55$ & $3.67$ & $1.06$ & $1.40$ & $1.00$\tabularnewline
$\beta_{71}$ & $1.63$ & $4.48$ & $1.00$ & $1.39$ & $1.01$\tabularnewline
$\beta_{81}$ & $1.63$ & $3.86$ & $1.15$ & $1.46$ & $1.00$\tabularnewline
$\beta_{91}$ & $1.56$ & $4.30$ & $1.03$ & $1.32$ & $1.00$\tabularnewline
$\beta_{101}$ & $1.55$ & $3.69$ & $1.07$ & $1.36$ & $1.00$\tabularnewline
$\tau_{1}^{2}$ & $30.54$ & $398.18$ & $59.13$ & $46.50$ & $28.69$\tabularnewline
$\rho_{\alpha}$ & $6.78$ & $83.99$ & $14.06$ & $10.29$ & $6.71$\tabularnewline
\hline
\end{tabular}
\end{table}

\begin{table}[H]
\caption{Comparison of Inefficiency Factors (IACT) for the Parameters with
Different Sampling Schemes (PG, data augmentation, MCMC-MH) of Mixed
Discrete-Linear Gaussian regression Simulation with random effects
$P=1000$ and $T=4$. \label{tab:Comparison-of-Different mixed discrete linear Gaussian2}}

\centering{}%
\begin{tabular}{cccccc}
\hline
Param.  & PG & MH1 & MH10 & MH20 & MH50\tabularnewline
\hline
$\beta_{12}$ & $1.51$ & $3.50$ & $1.15$ & $1.21$ & $1.00$\tabularnewline
$\beta_{22}$ & $1.59$ & $3.52$ & $1.07$ & $1.30$ & $1.00$\tabularnewline
$\beta_{32}$ & $1.53$ & $3.35$ & $1.01$ & $1.40$ & $1.00$\tabularnewline
$\beta_{42}$ & $1.59$ & $4.71$ & $1.00$ & $1.45$ & $1.00$\tabularnewline
$\beta_{52}$ & $1.53$ & $4.66$ & $1.00$ & $1.34$ & $1.00$\tabularnewline
$\beta_{62}$ & $1.61$ & $3.99$ & $1.01$ & $1.32$ & $1.00$\tabularnewline
$\beta_{72}$ & $1.52$ & $4.06$ & $1.00$ & $1.27$ & $1.00$\tabularnewline
$\beta_{82}$ & $1.50$ & $3.29$ & $1.00$ & $1.40$ & $1.00$\tabularnewline
$\beta_{92}$ & $1.50$ & $4.15$ & $1.01$ & $1.35$ & $1.00$\tabularnewline
$\beta_{102}$ & $1.58$ & $4.82$ & $1.01$ & $1.27$ & $1.00$\tabularnewline
$\tau_{2}^{2}$ & $1.42$ & $12.62$ & $1.63$ & $1.60$ & $1.49$\tabularnewline
$\rho$ & $7.24$ & $13.35$ & $9.85$ & $7.72$ & $7.56$\tabularnewline
\hline
\end{tabular}
\end{table}

\begin{table}[H]
\caption{TNV comparison of different sampling schemes (PG, data augmentation,
MCMC-MH) of Mixed Discrete-Linear Gaussian regression Simulation with
random effects $P=1000$ and $T=4$. \label{tab:Comparison-of-Different mixed discrete linear Gaussian2-1}}

\centering{}%
\begin{tabular}{cccccc}
\hline
 & PG & MH1 & MH10 & MH20 & MH50\tabularnewline
\hline
Time & $0.21$ & $0.18$ & $0.67$ &  $1.23$ & $2.89$\tabularnewline
$IACT_{mean}$ & $3.24$ & $24.46$ & $4.40$ & $3.90$ & $2.65$\tabularnewline
TNV & $0.68$ & $4.40$ & $2.95$ & $4.80$ & $7.66$\tabularnewline
rel. TNV & $1$ & $6.47$ & $4.34$ & $7.06$ & $11.26$\tabularnewline
\hline
\end{tabular}
\end{table}

\section{Additional Results on Simulation Bivariate Probit Regression}\label{tab:Comparison-of-Different bivariate probit}
This section provides additional simulation results on bivariate probit regression models in Section \ref{sub:TNV}
\begin{table}[H]
\caption{Comparison of Inefficiency Factors (IACT) for the Parameters with
Different Sampling Schemes (PG, data augmentation, MCMC-MH) of bivariate
Probit regression Simulation with random effects $P=1000$ and $T=4$
\label{tab:Comparison-of-Different biv probit1}}

\centering{}%
\begin{tabular}{ccccccc}
\hline
Param.  & PG & Data Aug. & MH1 & MH10 & MH20 & MH50\tabularnewline
\hline
$\beta_{11}$ & $1.00$ & $9.73$ & $2.13$ & $1.00$ & $1.00$ & $1.17$\tabularnewline
$\beta_{21}$ & $1.00$ & $11.59$ & $2.15$ & $1.00$ & $1.00$ & $1.19$\tabularnewline
$\beta_{31}$ & $1.00$ & $10.46$ & $1.00$ & $1.00$ & $1.00$ & $1.12$\tabularnewline
$\beta_{41}$ & $1.00$ & $11.01$ & $2.42$ & $1.01$ & $1.00$ & $1.19$\tabularnewline
$\beta_{51}$ & $1.00$ & $11.49$ & $2.82$ & $1.00$ & $1.00$ & $1.17$\tabularnewline
$\beta_{61}$ & $1.03$ & $9.91$ & $1.01$ & $1.00$ & $1.00$ & $1.16$\tabularnewline
$\beta_{71}$ & $1.00$ & $10.39$ & $2.49$ & $1.00$ & $1.00$ & $1.15$\tabularnewline
$\beta_{81}$ & $1.00$ & $11.40$ & $2.78$ & $1.00$ & $1.00$ & $1.20$\tabularnewline
$\beta_{91}$ & $1.01$ & $14.43$ & $2.24$ & $1.00$ & $1.00$ & $1.25$\tabularnewline
$\beta_{101}$ & $1.00$ & $10.44$ & $3.89$ & $1.00$ & $1.01$ & $1.20$\tabularnewline
$\tau_{1}^{2}$ & $18.20$ & $79.95$ & $178.23$ & $19.18$ & $20.17$ & $17.17$\tabularnewline
$\rho_{\alpha}$ & $13.91$ & $62.36$  & $92.64$ & $20.25$ & $15.44$ & $13.41$\tabularnewline
\hline
\end{tabular}
\end{table}

\begin{table}[H]
\caption{Comparison of Inefficiency Factors (IACT) for the Parameters with
Different Sampling Schemes (PG, data augmentation, MCMC-MH) of bivariate
Probit regression Simulation with random effects $P=1000$ and $T=4$
\label{tab:Comparison-of-Different biv probit2}}

\centering{}%
\begin{tabular}{ccccccc}
\hline
Param.  & PG & Data Aug. & MH1 & MH10 & MH20 & MH50\tabularnewline
\hline
$\beta_{12}$ & $1.00$ & $16.49$ & $1.02$ & $1.00$ & $1.00$ & $1.12$\tabularnewline
$\beta_{22}$ & $1.00$ & $17.03$ & $1.00$ & $1.09$ & $1.00$ & $1.07$\tabularnewline
$\beta_{32}$ & $1.00$ & $19.18$ & $4.32$ & $1.00$ & $1.04$ & $1.17$\tabularnewline
$\beta_{42}$ & $1.00$ & $23.78$ & $1.02$ & $1.00$ & $1.00$ & $1.15$\tabularnewline
$\beta_{52}$ & $1.00$ & $16.50$ & $5.52$ & $1.01$ & $1.03$ & $1.19$\tabularnewline
$\beta_{62}$ & $1.00$ & $21.58$ & $1.02$ & $1.00$ & $1.00$ & $1.15$\tabularnewline
$\beta_{72}$ & $1.00$ & $15.53$ & $1.00$ & $1.00$ & $1.00$ & $1.06$\tabularnewline
$\beta_{82}$ & $1.00$ & $17.08$ & 1.03  & $1.05$ & $1.00$ & $1.13$\tabularnewline
$\beta_{92}$ & $1.00$ & $15.37$ & $1.00$ & $1.00$ & $1.00$ & $1.11$\tabularnewline
$\beta_{102}$ & $1.01$ & $17.17$ & $2.79$ & $1.00$ & $1.00$ & $1.19$\tabularnewline
$\tau_{2}^{2}$ & $45.11$ & $207.69$ & $709.23$ & $102.65$ & $78.98$ & $47.92$\tabularnewline
$\rho$ & $9.69$ & $420.90$ & $14.29$ & $10.38$ & $9.09$ & $8.64$\tabularnewline
\hline
\end{tabular}
\end{table}

\section{Some Further Details of the Sampling Scheme for the Bivariate Probit Model with Random Effects\label{sec:Sampling-Scheme-for bivariate probit}}

Steps~3 and 4 of \Cref{alg: pmwg for biv probit} use a Hamiltonian Monte Carlo proposal to
sample $\bs \beta_{1}$ and $\bs \beta_{2}$ conditional
on the other parameters and random effects. The HMC requires
the gradient of $\log p\left(\boldsymbol{y}|\boldsymbol{\theta},\boldsymbol{\alpha}\right)$
with respect to $\boldsymbol{\beta}_{1}$ and $\bs \beta_2 $, where
\[
\frac{\partial\log p\left(\boldsymbol{y}|\boldsymbol{\theta},\boldsymbol{\alpha}\right)}{\partial\bs \beta_{1}}=\sum_{i=1}^{P}\sum_{t=1}^{T}\left(\frac{q_{1,it}g_{1,it}}{\Phi_{2}\left(w_{1,it},w_{2,it},q_{1,it}q_{2,it}\rho\right)}\right),
\]
\[
w_{1,it}=q_{1,it}\left(\boldsymbol{x}_{1,it}^{'}\boldsymbol{\beta}_{11}+\overline{\boldsymbol{x}}_{1,i}^{'}\boldsymbol{\beta}_{12}+\alpha_{1,i}\right),
\quad
w_{2,it}=q_{2,it}\left(\boldsymbol{x}_{2,it}^{'}\boldsymbol{\beta}_{21}+\overline{\boldsymbol{x}}_{2,i}^{'}\boldsymbol{\beta}_{22}+\alpha_{2,i}\right),
\]
and
\begin{eqnarray*}
g_{1,it} & = & \phi\left(w_{1,it}\right)\times\Phi\left(\frac{w_{2,it}-q_{1,it}q_{2,it}\rho w_{1,it}}{\sqrt{1-\left(q_{1,it}q_{2,it}\rho\right)^{2}}}\right).
\end{eqnarray*}
The gradient of $\log p\left(\boldsymbol{y}|\boldsymbol{\theta},\boldsymbol{\alpha}\right)$
with respect to $\boldsymbol{\beta}_{2}$ is obtained similarly.

\section{The Sampling Scheme for the Mixed Marginal Gaussian Regression
with Random Effects \label{sec:Sampling-Scheme-for mixed Gaussian}}

The model is described in \Cref{ss: mixed biv model with random effects}.
 The PG sampling scheme is similar to bivariate probit case and the following derivatives are needed in the HMC step.
From this section onwards, we denote $\boldsymbol{x}_{j,it}^{'}$
as all the covariates for the $j$th outcomes and $\eta_{j,it}=\boldsymbol{x}_{j,it}^{'}\boldsymbol{\beta}_{j}+\alpha_{j,i}$
for $j=1,2$
\[
\frac{\partial\log p\left(\boldsymbol{y}|\boldsymbol{\theta},\boldsymbol{\alpha}\right)}{\partial\boldsymbol{\beta}_{1}}=\sum_{i=1}^{P}\sum_{t=1}^{T}\left\{ \frac{\boldsymbol{x}_{1,it}y_{1,it}}{\sqrt{1-\rho^{2}}}\frac{\phi\left(\frac{\mu_{1|2}}{\sigma_{1|2}}\right)}{\Phi\left(\frac{\mu_{1|2}}{\sigma_{1|2}}\right)}-\frac{\boldsymbol{x}_{1,it}\left(1-y_{1,it}\right)}{\sqrt{1-\rho^{2}}}\frac{\phi\left(\frac{\mu_{1|2}}{\sigma_{1|2}}\right)}{1-\Phi\left(\frac{\mu_{1|2}}{\sigma_{1|2}}\right)}\right\}
\]
and
\begin{eqnarray*}
\frac{\partial\log p\left(\boldsymbol{y}|\boldsymbol{\theta},\boldsymbol{\alpha}\right)}{\partial\boldsymbol{\beta}_{2}} & = & \sum_{i=1}^{P}\sum_{t=1}^{T}-\frac{\rho}{\sqrt{1-\rho^{2}}}\boldsymbol{x}_{2,it}y_{1,it}\frac{\phi\left(\frac{\mu_{1|2}}{\sigma_{1|2}}\right)}{\Phi\left(\frac{\mu_{1|2}}{\sigma_{1|2}}\right)}\\
 &  & +\boldsymbol{x}_{2,it}y_{1,it}\left(y_{2,it}-\eta_{2,it}\right)\\
 &  & +\frac{\left(1-y_{1,it}\right)}{1-\Phi\left(\frac{\mu_{1|2}}{\sigma_{1|2}}\right)}\phi\left(\frac{\mu_{1|2}}{\sigma_{1|2}}\right)\left(\frac{\rho}{\sqrt{1-\rho^{2}}}\boldsymbol{x}_{2,it}\right)\\
 &  & +\boldsymbol{x}_{2,it}\left(1-y_{1,it}\right)\left(y_{2,it}-\eta_{2,it}\right)
\end{eqnarray*}

\section{Gradients for the HMC Sampling Scheme for the Mixed Marginal Clayton Copula Regression
with Random Effects \label{sec:Sampling-Scheme-for mixed clayton}}
\begin{eqnarray*}
\frac{\partial\log p\left(\boldsymbol{y}|\boldsymbol{\theta},\boldsymbol{\alpha}\right)}{\partial\boldsymbol{\beta}_{1}} & = & \sum_{i=1}^{P}\sum_{t=1}^{T}\left\{ \frac{\boldsymbol{x}_{1,it}y_{1,it}}{1-C_{1|2}^{\rm Cl}\left(u_{1,it}|u_{2,it}\right)}\left(u_{2,it}^{-\theta-1}\left(-1-\frac{1}{\theta}\right)
\left(\Phi\left(-\eta_{1,it}\right)^{-\theta}+u_{2,it}^{-\theta}-1\right)^{-2-\frac{1}{\theta}}\right)\right.\\
 &  & \phi\left(-\eta_{1,it}\right)\left(-\theta\right)\Phi\left(-\eta_{1,it}\right)^{-\theta-1}\frac{\boldsymbol{x}_{1,it}\left(1-y_{1,it}\right)}
 {C_{1|2}^{\rm Cl}\left(u_{1,it}|u_{2,it}\right)}\phi\left(-\eta_{1,it}\right)\left(\theta\right)\Phi\left(-\eta_{1,it}\right)^{-\theta-1}\\
 &  & \left.\left(u_{2,it}^{-\theta-1}\left(-1-\frac{1}{\theta}\right)\left(\Phi\left(-\eta_{1,it}\right)^{-\theta}+u_{2,it}^{-\theta}-1\right)^{-2-\frac{1}
 {\theta}}\right)\right\};
\end{eqnarray*}
\[
\frac{\partial\log p\left(\boldsymbol{y}|\boldsymbol{\theta},\boldsymbol{\alpha}\right)}{\partial\boldsymbol{\beta}_{2}}=\sum_{i=1}^{P}\sum_{t=1}^{T}\left\{ I_{it}+II_{it}+III_{it}+IV_{it}\right\} ,
\]
where
\begin{align*}
I_{it}& \coloneqq\frac{\boldsymbol{x}_{2,it}y_{1,it}}{1-C_{1|2}^{\rm Cl}\left(u_{1,it}|u_{2,it}\right)}\left(\nabla u v+u \nabla v \right),
\quad
\zeta_{it}=y_{2,it}-\eta_{2,it}, \\
u & :=-\Phi\left(\zeta_{it}\right)^{-\theta-1} \quad
\nabla {u} \coloneqq\frac{du}{d\boldsymbol{\beta}_{2}}=\left(-1-\theta\right)\Phi\left(\zeta_{it}\right)^{-\theta-2}\phi\left(\zeta_{it}\right),
v=\left(u_{1,it}^{-\theta}+\Phi\left(\zeta_{it}\right)^{-\theta}-1\right)^{-1-\frac{1}{\theta}}, \\
\nabla {v} &\coloneqq\frac{dv}{d\boldsymbol{\beta}_{2}}=\left(-1-\frac{1}{\theta}\right)\left(u_{1,it}^{-\theta}+\Phi\left(\zeta_{it}\right)^{-\theta}
-1\right)^{-2-\frac{1}{\theta}}\theta\Phi\left(\zeta_{it}\right)^{-\theta-1}\phi\left(\zeta_{it}\right)
\end{align*}
The terms $II_{it}$, $III_{it}$, and $IV_{it}$ are given by
\begin{eqnarray*}
II_{it} & = & \boldsymbol{x}_{2,it}y_{1,it}\left(y_{2,it}-\eta_{2,it}\right),
III_{it}  =  \boldsymbol{x}_{2,it}\left(1-y_{1,it}\right)\left(y_{2,it}-\eta_{2,it}\right),\\
IV_{it} & = & \frac{\boldsymbol{x}_{2,it}\left(1-y_{1,it}\right)}{C_{1|2}^{Cl}\left(u_{1,it}|u_{2,it}\right)}\left(\nabla {u}_{IV}v_{IV}+u_{IV}\nabla {v}_{IV}\right),
\end{eqnarray*}
where
\begin{align*}
u_{IV} & =\Phi\left(\zeta_{it}\right)^{-\theta-1},\nabla{u}_{IV}=-\left(-\theta-1\right)\Phi\left(\zeta_{it}\right)^{-\theta-2}\phi\left(\zeta_{it}\right),
v_{IV}=\left(u_{1,it}^{-\theta}+\Phi\left(\zeta_{it}\right)^{-\theta}-1\right)^{-1-\frac{1}{\theta}},\\
\nabla{v}_{IV}& =\left(-1-\frac{1}{\theta}\right)
\left(u_{1,it}^{-\theta}+\Phi\left(\zeta_{it}\right)^{-\theta}-1\right)^{-2-\frac{1}{\theta}}\theta\Phi\left(\zeta_{it}\right)^{-\theta-1}
\phi\left(\zeta_{it}\right).
\end{align*}
The PG sampling scheme is similar to that for the bivariate probit
case except that in step 2 we work with the unconstrained parameter $\theta_{\rm un} = \log \theta$, where $\theta> 0 $ for the Clayton copula.

\section{Gradients for the HMC Sampling Scheme for the Mixed Marginal Gumbel Copula Regression
with Random Effects \label{sec:Sampling-Scheme-for mixed gumbel}}

\begin{eqnarray*}
\frac{\partial\log p\left(\boldsymbol{y}|\boldsymbol{\theta},\boldsymbol{\alpha}\right)}{\partial\boldsymbol{\beta}_{1}} & = & \sum_{i=1}^{P}\sum_{t=1}^{T}(I_{it}+II_{it}),
\end{eqnarray*}
\begin{align*}
I_{it}& =\frac{\boldsymbol{x}_{1,it}y_{1,it}}{1-C_{1|2}^{\rm Gu}\left(u_{1,it}|u_{2,it}\right)}\left(-\frac{1}{u_{2,it}}\right)\left(-\log u_{2,it}\right)^{\theta-1}\left(\nabla{u}v+u\nabla{v}\right),\\
II_{it} & =\frac{\boldsymbol{x}_{1,it}\left(1-y_{1,it}\right)}{C_{1|2}\left(u_{1,it}|u_{2,it}\right)}\left(\frac{1}{u_{2,it}}\right)\left(-\log u_{2,it}\right)^{\theta-1}\left(\nabla{u}v+u\nabla{v}\right),
\intertext{where}
u& = \exp\left(-\left(\left(-\log u_{1,it}\right)^{\theta}+\left(-\log u_{2,it}\right)^{\theta}\right)^{1/\theta}\right)\\
\nabla{u} & =  \exp\left(-\left(\left(-\log u_{1,it}\right)^{\theta}+\left(-\log u_{2,it}\right)^{\theta}\right)^{1/\theta}\right)\left(-\frac{1}{\theta}\right)\left(\left(-\log u_{1,it}\right)^{\theta}+\left(-\log u_{2,it}\right)^{\theta}\right)^{\frac{1}{\theta}-1}\\
 &   \theta\left(-\log u_{1,it}\right)^{\theta-1}\left(\frac{1}{u_{1,it}}\right)\phi\left(-\eta_{1}\right),
v=\left(\left(-\log u_{1,it}\right)^{\theta}+\left(-\log u_{2,it}\right)^{\theta}\right)^{\frac{1}{\theta}-1}\\
\nabla{v} & =  \left(\frac{1}{\theta}-1\right)\left(\left(-\log u_{1,it}\right)^{\theta}+\left(-\log u_{2,it}\right)^{\theta}\right)^{\frac{1}{\theta}-2}
  \theta\left(-\log u_{1,it}\right)^{\theta-1}\left(\frac{1}{u_{1,it}}\right)\phi\left(-\eta_{1,it}\right).
\end{align*}

\begin{align*}
\frac{\partial\log p\left(\boldsymbol{y}|\boldsymbol{\theta},\boldsymbol{\alpha}\right)}{\partial\boldsymbol{\beta}_{2}} & =\sum_{i=1}^{P}\sum_{t=1}^{T}\left\{ I_{it}+II_{it}+III_{it}+IV_{it}\right\} ,\\
I_{it} &=\frac{\boldsymbol{x}_{2,it}y_{1,it}}{1-C_{1|2}^{\rm Gu}\left(u_{1}|u_{2}\right)}\left(-\left(\nabla{u}vwz+u\nabla{v}wz+uv\nabla{w}z+uvw\nabla{z}\right)\right),\\
II_{it} &=\boldsymbol{x}_{2,it}y_{1,it}\left(y_{2,it}-\left(\eta_{2,it}\right)\right),\\
III_{it} &=\boldsymbol{x}_{2,it}\left(1-y_{1,it}\right)\left(y_{2,it}-\left(\eta_{2,it}\right)\right),\\
IV_{it} &=\frac{\boldsymbol{x}_{2,it}\left(1-y_{1,it}\right)}{C_{1|2}^{Gu}\left(u_{1,it}|u_{2,it}\right)}
\left(\nabla{u}vwz+u\nabla{v}wz+uv\nabla{w}z+uvw\nabla{z}\right),
\intertext{where}
\zeta_{it}& =\left(y_{2,it}- \left(\eta_{2,it}\right)\right),
u =\frac{1}{u_{2,it}},
\nabla{u}:=\frac{\partial u}{\partial \bs \beta_2} = u_{2,it}^{-2}\phi\left(\zeta_{it}\right),
v=\left(-\log\Phi\left(\zeta_{it}\right)\right)^{\theta-1}\\
\nabla{v}& := \frac{\partial u}{\partial \bs \beta_2}= \left(\theta-1\right)\left(-\log\Phi\left(\zeta_{it}\right)\right)^{\theta-2}\frac{1}
{\Phi\left(\zeta_{it}\right)}\phi\left(\zeta_{it}\right)\\
w & =\exp\left(-\left(\left(-\log u_{1,it}\right)^{\theta}+\left(-\log u_{2,it}\right)^{\theta}\right)^{1/\theta}\right), \\
\nabla{w}&:= \frac{\partial u}{\partial \bs \beta_2}= \exp\left(-\left(\left(-\log u_{1,it}\right)^{\theta}+\left(-\log u_{2,it}\right)^{\theta}\right)^{1/\theta}\right)\left(-\frac{1}{\theta}\right)\left(\left(-\log u_{1,it}\right)^{\theta}+\left(-\log u_{2,it}\right)^{\theta}\right)^{\frac{1}{\theta}-1}\\
 &  \theta\left(-\log u_{2,it}\right)^{\theta-1}\left(\frac{1}{u_{2,it}}\right)\phi\left(\zeta_{it}\right)\\
z  &=\left(\left(-\log u_{1,it}\right)^{\theta}+\left(-\log u_{2,it}\right)^{\theta}\right)^{\frac{1}{\theta}-1} \\
\nabla{z} & =  \left(\frac{1}{\theta}-1\right)\left(\left(-\log u_{1,it}\right)^{\theta}+\left(-\log u_{2,it}\right)^{\theta}\right)^{\frac{1}{\theta}-2}
  \theta\left(-\log u_{2,it}\right)^{\theta-1}\left(\frac{1}{u_{2,it}}\right)\phi\left(\zeta_{it}\right)
\end{align*}

The PG sampling scheme is similar to that for the bivariate probit
model case except that in step 2 we reparametrize the Gumbel dependence parameter to
$\theta_{\rm un}:= \log (\theta -1 )$ because $\theta > 1$.

\section{Data Augmentation Bivariate Probit Models with Random Effects\label{sec:Data-Augmentation-Bivariate probit}}

We will work with the augmented posterior distribution
\begin{eqnarray*}
p\left(\boldsymbol{y}^{*},\left\{ \boldsymbol{\alpha}_{i}\right\} ,\boldsymbol{\theta}|\boldsymbol{y}\right) & = & p\left(\boldsymbol{y}|\boldsymbol{y}^{*},\left\{ \boldsymbol{\alpha}_{i}\right\} ,\boldsymbol{\theta}\right)p\left(\boldsymbol{y}^{*}|\left\{ \boldsymbol{\alpha}_{i}\right\} ,\boldsymbol{\theta}\right)p\left(\left\{ \boldsymbol{\alpha}_{i}\right\} |\boldsymbol{\theta}\right)p\left(\boldsymbol{\theta}\right),
\end{eqnarray*}
where
\begin{eqnarray*}
p\left(\boldsymbol{y}|\boldsymbol{y}^{*},\left\{ \boldsymbol{\alpha}_{i}\right\} ,\boldsymbol{\theta}\right) & = & \prod_{i=1}^{P}\prod_{t=1}^{T}\left[I\left(y_{1,it}^{*}\leq0\right)I\left(y_{1,it}=0\right)+I\left(y_{1,it}^{*}>0\right)I\left(y_{1,it}=1\right)\right]\\
 &  & \left[I\left(y_{2,it}^{*}\leq0\right)I\left(y_{2,it}=0\right)+I\left(y_{2,it}^{*}>0\right)I\left(y_{2,it}=1\right)\right],
\end{eqnarray*}
\[
p\left(\boldsymbol{y}^{*}|\left\{ \boldsymbol{\alpha}_{i}\right\} ,\boldsymbol{\theta}\right)\sim N\left(\bs \mu_{it}, \bs \Sigma_\epsilon \right ) ,
p\left(\left\{ \boldsymbol{\alpha}_{i}\right\} |\boldsymbol{\theta}\right)\sim N \left (\bs 0, \bs \Sigma_\alpha \right )
\]
where $\boldsymbol{\mu}_{it}=\left(\boldsymbol{\eta}_{1,it},\boldsymbol{\eta}_{2,it}\right)^{T}$, $\bs \Sigma_\epsilon $ and $\bs \Sigma_\alpha$ are defined in Equations \eqref{eq: cov matrices} in Section \ref{ss: biv probit with random effects},
and $p\left(\boldsymbol{\theta}\right)$ is the prior distribution
for $\boldsymbol{\theta}$. The following complete conditional posteriors
for the Gibbs sampler can then be derived.

\begin{equation}
y_{1,it}^{*}|\left\{ \boldsymbol{\alpha}_{i}\right\} ,\boldsymbol{\theta},\boldsymbol{y},y_{2,it}^{*}\sim\begin{cases}
TN_{\left(-\infty,0\right)}\left(\mu_{1|2},\sigma_{1|2}\right), & y_{1,it}=0\\
TN_{\left(0,\infty\right)}\left(\mu_{1|2},\sigma_{1|2}\right), & y_{1,it}=1
\end{cases}
\end{equation}
and
\begin{equation}
y_{2,it}^{*}|\left\{ \boldsymbol{\alpha}_{i}\right\} ,\boldsymbol{\theta},\boldsymbol{y},y_{1,it}^{*}\sim\begin{cases}
TN_{\left(-\infty,0\right)}\left(\mu_{2|1},\sigma_{2|1}\right), & y_{2,it}=0\\
TN_{\left(0,\infty\right)}\left(\mu_{2|1},\sigma_{2|1}\right), & y_{2,it}=1
\end{cases},
\end{equation}
where $\mu_{1|2}=\left(\boldsymbol{x}_{1,it}^{'}\boldsymbol{\beta}_{1}+\alpha_{1,i}\right)+\rho\left(y_{2,it}^{*}-\left(\boldsymbol{x}_{2,it}^{'}\boldsymbol{\beta}_{2}+\alpha_{2,i}\right)\right)$,
and $\sigma_{1|2}=\sqrt{1-\rho^{2}}$ and $\mu_{2|1}$ and $\sigma_{2|1}$
are defined similarly;
\begin{equation}
\boldsymbol{\alpha}_{i}|\boldsymbol{y}^{*},\boldsymbol{y},\Sigma_{\alpha},\Sigma,\boldsymbol{\beta}\sim N\left(\bs D_{\alpha_{i}}\bs d_{\alpha_{i}},\bs D_{\alpha_{i}}\right),
\end{equation}
where $\bs D_{\alpha_{i}}=\left(\bs T\bs \Sigma^{-1}+\bs \Sigma_{\alpha}^{-1}\right)^{-1}$
and $\bs d_{\alpha_{i}}=\bs \Sigma^{-1}\sum_{t}\left(\left[\begin{array}{c}
y_{1,it}^{*}\\
y_{2,it}^{*}
\end{array}\right]-\left[\begin{array}{cc}
\boldsymbol{x}_{1,it}^{'} & \boldsymbol{0}\\
\boldsymbol{0} & \boldsymbol{x}_{2,it}^{'}
\end{array}\right]\left[\begin{array}{c}
\boldsymbol{\beta}_{1}\\
\boldsymbol{\beta}_{2}
\end{array}\right]\right)$ for $i=1,...,P$;
\begin{equation}
\boldsymbol{\beta}|\boldsymbol{y}^{*},\boldsymbol{y},\bs \Sigma_{\alpha},\bs \Sigma,\left\{ \alpha_{i}\right\} \sim N\left(\bs D_{\beta}\bs d_{\beta},\bs D_{\beta}\right),
\end{equation}
where
\[
\bs D_{\bs \beta}=\left(\sum_{i=1}^{P}\sum_{t=1}^{T}\left[\begin{array}{cc}
\boldsymbol{x}_{1,it} & \boldsymbol{0}\\
\boldsymbol{0} & \boldsymbol{x}_{2,it}
\end{array}\right]^{'}\bs \Sigma^{-1}\left[\begin{array}{cc}
\boldsymbol{x}_{1,it} & \boldsymbol{0}\\
\boldsymbol{0} & \boldsymbol{x}_{2,it}
\end{array}\right]+\bs \Sigma_{0}^{-1}\right)^{-1},
\]
 and

\[
\bs d_{\beta}=\sum_{i=1}^{P}\sum_{t=1}^{T}\left[\begin{array}{cc}
\boldsymbol{x}_{1,it} & \boldsymbol{0}\\
\boldsymbol{0} & \boldsymbol{x}_{2,it}
\end{array}\right]^{'}\Sigma^{-1}\left(\left[\begin{array}{c}
y_{1,it}^{*}\\
y_{2,it}^{*}
\end{array}\right]-\left[\begin{array}{c}
\alpha_{1,i}\\
\alpha_{2,i}
\end{array}\right]\right)
\]

\begin{equation}
\bs \Sigma_{\alpha}|\boldsymbol{y}^{*},\boldsymbol{y},\bs \Sigma,\boldsymbol{\beta}\sim W\left(v_{1},\bs R_{1}\right),
\end{equation}
where
\[
v_{1}=v_{0}+P,
\quad \text{and} \quad
\bs R_{1}=\left[\bs R_{0}^{-1}+\sum_{i=1}^{P}\left[\begin{array}{c}
\alpha_{1,i}\\
\alpha_{2,i}
\end{array}\right]\left[\begin{array}{cc}
\alpha_{1,i} & \alpha_{2,i}\end{array}\right]\right]^{-1}.
\]
We use Metropolis within Gibbs to sample $\rho$.

\end{document}